\documentclass[a4paper, reqno, 11pt]{amsart}
\usepackage{color}
\usepackage{amsmath}
\usepackage{amsthm}
\usepackage{amssymb}
\usepackage{amsmath}
\usepackage{mathcomp}
\usepackage{dsfont}
\usepackage[toc,page]{appendix}
\usepackage{graphicx} 
\usepackage{subfigure}
\usepackage{enumerate}
\usepackage{amsfonts}
\usepackage{amsthm}
\usepackage{amsmath}
\usepackage{times}
\usepackage{amscd}
\usepackage[latin2]{inputenc}
\usepackage{t1enc}
\usepackage{enumitem}
\usepackage[mathscr]{eucal}
\usepackage{indentfirst}
\usepackage{graphicx}
\usepackage{graphics}
\usepackage{pict2e}
\usepackage{epic}
\usepackage{esint}
\usepackage{epstopdf} 
\usepackage{verbatim}
\usepackage{cancel}
\usepackage{amssymb}
\usepackage{xcolor}
\usepackage{dsfont}
\usepackage{hyperref}
\usepackage{physics}
\allowdisplaybreaks

\definecolor{Greenish}{RGB}{34, 139 , 34}
\definecolor{Blueish}{RGB}{39,64, 139}

\usepackage{hyperref}
\hypersetup{
    colorlinks=true,
    linkcolor=Greenish,
    citecolor=Blueish,
    }

\usepackage[margin=30mm]{geometry}
    
\newtheorem{theorem}{Theorem}[section]
\newtheorem{lemma}[theorem]{Lemma}

\newtheorem{corollary}[theorem]{Corollary}

\theoremstyle{definition}
 
\newtheorem{definition}{Definition}[section]
\newtheorem{remark}[definition]{Remark}

\numberwithin{equation}{section}

\def\b1{\mathds{1}}
\def\1{\mathds{1}}

\def\Re{\mathrm{Re}}

\def\a0{\mathfrak{a}}

\newcommand{\vertiii}[1]{{\left\vert\kern-0.25ex\left\vert\kern-0.25ex\left\vert #1 
    \right\vert\kern-0.25ex\right\vert\kern-0.25ex\right\vert}}

\begin{document}
\title[Generating function for quantum depletion of Bose-Einstein condensates]{Generating function for quantum depletion of Bose-Einstein condensates}

\author[S. Rademacher]{Simone Rademacher}
\address{Department of Mathematics, LMU Munich, Theresienstrasse 39, 80333 Munich, Germany} 
\email{simone.rademacher@math.lmu.de}

\date{\today}

\begin{abstract} We consider a Bose gas on the unit torus at zero temperature in the Gross-Pitaevskii regime, known to perform Bose-Einstein condensation: a macroscopic fraction of the bosons occupy the same quantum state, called condensate. We study the Bose gas' quantum depletion, that is the number of bosons outside the condensate, and derive an explicit asymptotic formula of its generating function. Moreover, we prove an upper bound for the tails of the quantum depletion. 
\end{abstract}

\maketitle
\section{Introduction and Results}

\subsection{Introduction}
 
As predicted by Bose \cite{Bose-24} and Einstein \cite{Einstein-24}, and later also observed experimentally \cite{WC-95,K-95}, trapped Bose gases show a peculiar phase transition at extremely low temperatures: a macroscopic fraction of the particle condense into the same one-particle quantum state, called Bose-Einstein condensate. This paper is dedicated to the mathematical description of Bose-Einstein condensates. 

For this, we consider $N$ bosons on the unit torus $\pi = \mathbb{T}^3$ in three dimensions described on $L_s^2( \pi^N)$, the symmetric subspace of $L^2( \pi^N)$ through the  Hamiltonian 
\begin{align}
\label{def:ham}
H_N = \sum_{i=1}^N ( - \Delta_i ) + \frac{1}{N} \sum_{i=1}^N V_N (x_i - x_j)   \; . 
\end{align}
where in the following, we assume that the two-body interaction $V$ is an element of $L^3( \pi)$, non-negative, compactly supported and spherically symmetric. The two-body interaction $V_N (x) = N^3 V (N x)$ depends on the distance of the $i$-th and $j$-th particle on the torus $\mathbb{T}^3$ only, and scales with the total number of particles $N$ and models approximate delta-interaction in the large particle limit and corresponds to the Gross-Pitaevskii regime. At zero temperature, the Bose gas relaxes to the ground state. In fact under the conditions on $V$ formulated above, the Hamiltonian $H_N$ is on the domain of smooth functions well defined, bounded by below and, by Friedrich's method, extendable to a self-adjoint operator (that, in abuse of notation, we call $H_N$ in the following, too). The unique (up to complex phase) ground state $\psi_N$ of $H_N$ is well known to exhibit Bose-Einstein condensation: a macroscopic fraction of the particles condense into the same (one-particle) quantum state, the so-called Bose-Einstein condensate, in our setting given by $\varphi = 1_\pi \in L^2( \pi)$. Mathematically, the property of Bose-Einstein condensation is formulated in terms of the quantum depletion, that is the operator
\begin{align}
\mathcal{N}_+ = \sum_{i=1}^N Q_i , \quad \text{with} \quad Q =  1- \vert \varphi \rangle \langle \varphi \vert  \label{def:Q}
\end{align}
and where $Q_i$ is the operator that acts as the projection $Q$ on the $i$-th particle. Thus, by definition \eqref{def:Q}, the quantum depletion counts the number of particles outside of the condensate. Then, the ground state $\psi_N$ of \eqref{def:ham} is said to satisfy Bose-Einstein condensation, if 
\begin{align}
\frac{\mathbb{E}_{\psi_N} \big[ \mathcal{N}_+  \big]}{N} =& \frac{\langle \psi_N, \mathcal{N}_+ \psi_N \rangle}{N} \rightarrow 0 \quad \text{as} \quad N \rightarrow \infty \; . \label{eq:BEC}
\end{align}
The property of Bose-Einstein condensation \label{eq:BEC} was first proven by \cite{LS-02}. Later, the rate of convergence of \eqref{eq:BEC} was studied in great detail since then \cite{BCCS_cond,BCCS,BCCS_optimal,H,HST}  and showed  to be $O(1/N)$. In the past decades, the property of Bose-Einstein condensation and excitations beyond the condensate have been studied in various settings: we refer to \cite{LS,NRS,NT,BSS_bogo, BSS_optimal,NNRT} for generalizations for Bose gases on $\mathbb{R}^3$ trapped through an external potential. 

The analysis of the number of excitations beyond the condensate, counted by the operator $\mathcal{N}_+$, is based on a mathematical verification of Bogoliubov's theory \cite{Bogoliubov-47} on the Bose gas' excitation spectrum given in \cite{BCCS}. To be more precise \cite{BCCS} proves an asymptotic formula ofthe expectation of $\mathcal{N}_+$ in the large particle limit given in terms of the interaction potential $V$'s scattering length, defined through the solution $f$ of the potential's scattering equation by 
\begin{align}
\a0 = \int V(x) f(x) dx \quad \text{where} \quad \bigg( - \Delta +  \frac{1}{2} V \bigg) f =0, \quad \text{and} \quad f(x) \xrightarrow{ \vert x \vert \rightarrow \infty} 1  \; . 
\end{align}
To be more precise, \cite{BCCS} shows that the expectation of the quantum depletion is asymptotically 
\begin{align}
\label{def:mu}
\mu := \lim_{N \rightarrow \infty} \mathbb{E}_{\psi_N} \big[ \mathcal{N}_+ \big] = \sum_{p \in   \pi_+^*} \sinh^2 ( \nu_p) 
\end{align}
with 
\begin{align}
\label{def:nu}
\nu_p := \frac{1}{4} \log \bigg( \frac{p^2}{p^2 + 16 \pi \a0} \bigg) , \quad \text{and}\quad  \pi^*_+ = (2 \pi \mathbb{Z})^3 \setminus \lbrace 0 \rbrace 
\end{align}
and thus, in particular, in the large particle limit $\mathcal{O}(1)$. Note that the mean $\lambda$ in \eqref{def:mu} is defined in momentum space $\pi^*$, as in our setting it turns out to be more convenient to work in momentum instead of position space.

\subsection{Results} We further improve the characterization of the Bose gas' quantum depletion and compute an asymptotic formula for the generating function. 

\begin{theorem}
\label{thm:mgf}
Let $V \in L^3( \pi)$ be non-negative, compactly supported and spherically symmetric and $\psi_N$ denote the ground state of $H_N$ defined in \eqref{def:ham}. Then there exists $\lambda_0 >0$ (depending on $V$ only), for all $ \vert \lambda \vert \leq \lambda_0 $, we have 
\begin{align}
\label{eq:thm-mgf}
\lim_{N \rightarrow \infty} \mathbb{E}_{\psi_N} \big[ e^{\lambda \mathcal{N}_+ } \big] =  e^{\Lambda ( \lambda ) } \; 
\end{align}
where $\Lambda : [-\lambda_0, \lambda_0] \rightarrow \mathbb{R}$ is a convex function given by 
\begin{align}
\label{def:Lambda}
\Lambda ( \lambda ) =&  \int_0^\lambda \sum_{p \in \pi_+^* } \cosh^2(\nu_p) \sinh^2( \nu_p) \frac{  2 \cosh^2( \nu_p) (\cosh( 2\kappa) - 1)  - e^{-2\kappa} + 1}{  1 -2 \cosh^2(\nu_p) \sinh^2( \nu_p) (\cosh(2\kappa) - 1)  }  d\kappa  \notag \\
&+ \lambda \sum_{p \in \pi_+^*} \sinh^2 (\nu_p)  \; . 
\end{align}
Moreover, for $k \in \mathbb{N}$, we have 
\begin{align}
\lim_{N \rightarrow \infty}\mathbb{E}_{\psi_N} \big[ \mathcal{N}_+^k \big] = \frac{d^k}{d\lambda^k} \; e^{  \Lambda ( \lambda) } \bigg\vert_{\lambda =0} \; . \label{eq:Nhochk}
\end{align}
\end{theorem}

\begin{remark}
We collect some remarks on Theorem \ref{thm:mgf}. 
\begin{enumerate}
\item[(i)] First we note that the r.h.s. of \eqref{def:Lambda} is finite for sufficiently small $\lambda_0>0$, since  $ \sinh ( \nu_p ) \in \ell^2( \pi_+^*)$ by definition of $\nu_p$ in \eqref{def:nu}. 
\item[(ii)] In the proof of Theorem \ref{thm:mgf} we establish a rate of convergence of \eqref{eq:thm-mgf} that is $\mathcal{O} ( N^{-1/4})$ (see \eqref{eq:mfg-approx}).
\item[(iii)] In Theorem \ref{thm:mgf}, more precisely \eqref{eq:Nhochk}, we prove that the asymptotics of $\mathbb{E}_{\psi_N} \big[ e^{\lambda \mathcal{N}_+ }\big]$ can indeed be understood as a generating function of the asymptotic number of excitations. Through \eqref{eq:Nhochk} computations of any moment of the number of excitations $\mathcal{N}_+$ reduces to taking derivatives of the function $\Lambda$ defined in \eqref{def:Lambda}. 
\item[(iv)] In the proof, we embed the problem in the bosonic Fock space (see Section \ref{sec:prel} and Section \ref{sec:proofs} for more details) where we can write the operator $\mathcal{N}_+$ in terms of the creation and annihilation operators as  $\mathcal{N}_+ = \sum_{p \in \pi_+^*} a_p^*a_p$ with $\pi_+^* = 2 \pi \mathbb{Z}^3 \setminus \lbrace 0 \rbrace$ . It turns out that Theorem \ref{thm:mgf} then easily generalizes to operators of the form $ \sum_{p \in \pi_+^*} \tau_p a_p^*a_p$ for sequences $\tau_p \in \ell^2( \pi_+^*)$ such that $\cosh( 2 \lambda \tau_p) - 1 < 1/(2 \sinh^2( \nu_p) \cosh^2( \nu_p))$ for all $p \in \pi_+^*$ and all $\vert \lambda \vert \leq \lambda_0$. Then, the corresponding moment generating function satisfies 
\begin{align}
	\lim_{N \rightarrow \infty} \mathbb{E}_{\psi_N} \bigg[ e^{\lambda \sum_{p \in \pi_+^*} \tau_p a_p^*a_p } \bigg] =  e^{\widetilde{\Lambda} ( \lambda ) }
\end{align}
where 
\begin{align}
\label{def:Lambda}
\widetilde{\Lambda} ( \lambda ) =&  \int_0^\lambda \sum_{p \in \pi_+^* } \tau_p \cosh^2(\nu_p) \sinh^2( \nu_p) \frac{  2 \cosh^2( \nu_p) (\cosh( 2\kappa \tau_p) - 1)  - e^{-2\kappa \tau_p} + 1}{  1 -2 \cosh^2(\nu_p) \sinh^2( \nu_p) (\cosh(2\kappa \tau_p) - 1)  }  d\kappa  \notag \\
&+ \lambda \sum_{p \in \pi_+^*} \tau_p \sinh^2 (\nu_p)  \; . 
\end{align}
\end{enumerate}
\end{remark}

\subsubsection*{Exponential and Moment bounds} Theorem \ref{thm:mgf} provides a detailed description of the number of excitations through an asymptotic formula of the generating function. As an immediate consequence, we recover back \eqref{def:mu} and, moreover, any moment of the number of excitation is bounded by a constant independent in the total number of particles $N$, i.e. 
\begin{align}
\mathbb{E} \big[ \mathcal{N}_+^k  \big] \leq C_k \label{eq:moments}
\end{align}
for $C_k >0$ and, moreover 
\begin{align}
\mathbb{E} \big[ e^{\lambda \mathcal{N}_+} \big] \leq C_1 e^{C_2 \lambda} \; .\label{eq:expbounds}
\end{align}
for small $\lambda >0$ and $C_1,C_2 >0$. With \eqref{eq:moments} and \eqref{eq:expbounds} we recover back earlier results (see \cite{BCCS} for \eqref{eq:moments} resp. \cite{NR} for \eqref{eq:expbounds}). Note, however, that the proof of Theorem \ref{thm:mgf} relies on the apriori bound exponential of the form \eqref{eq:exp-bound} from \cite{NR}.

\subsubsection*{Limiting distribution}  The existence of an asymptotic formula for the characteristic function  $\mathbb{E} \big[ e^{{\rm i} \lambda\mathcal{N}_+} \big]$, where $\lambda \in \mathbb{R}$ and ${\rm i}$ denotes the imaginary unit, implies, by Levy's continuity theorem, that the random variable $\mathcal{N}_+$ has in distribution an asymptotic limit. We note that Theorem \ref{thm:mgf} establishes an asymptotic formula for the moment generating function $\mathbb{E} \big[ e^{\lambda \mathcal{N}_+} \big]$ only for real values of $\lambda$, since its proof relies on the solution of an ordinary differential equation defined over the real numbers. Nonetheless, we conjecture that the asymptotic behavior also holds for purely imaginary values of $\lambda$, and thus that the random variable $\mathcal{N}_+$ has in distribution an asymptotic limit for $N \rightarrow \infty$. 



\subsubsection*{Characterization of tails} As a consequence of Theorem \ref{thm:mgf} (more specifically \eqref{eq:Nhochk}), the deviation of $\mathcal{N}_+$ from its asymptotic expectation value $\mu$, defined in \eqref{def:mu}, is in the large particle limit 
\begin{align}
\label{def:sigma}
 \sqrt{\mathbb{E}_{\psi_N} \big[ \big( \mathcal{N}_+ - \mu \big)^2 \big]} \xrightarrow[N \rightarrow \infty]{} \sigma, \; \quad \text{with} \quad \sigma^2 := 2  \sum_{p \in \pi_+^*} \sinh^2( \nu_p) \cosh^2 ( \nu_p)  
\end{align}
(see Section \ref{sec:lln-lde} for the proof of \eqref{def:sigma}). In particular, the deviation does not vanish in the limit $N \rightarrow \infty$. Related to that we show in the following Corollary that $\mathcal{N}_+$ does not converge in probability to its asymptotic expectation value $\mu = \lim_{N \rightarrow \infty} \mathbb{E} \big[ \mathcal{N}_+ \big]$. This statement is complemented by an asymptotic upper bound for the tails of the number of excitations.

\begin{corollary}
\label{cor:lde}
Under the same assumptions as in Theorem \ref{thm:mgf}, assume that $\a0 \not= 0$. Then, there exists $n, \varepsilon_n >0$ such that 
\begin{align}
\label{eq:lln}
\liminf_{N \rightarrow \infty }\mathbb{P}_{\psi_N} \big[  \big\vert \mathcal{N}_+ - \mu  \big\vert > n \big]  > \varepsilon_n \; ,  
\end{align}
where $\mu$ is defined by \eqref{def:mu}, and 
\begin{align}
\limsup_{N \rightarrow \infty}  \mathbb{P}_{\psi_N} \big[ \mathcal{N}_+ - \mu \geq n \big] \leq e^{\inf_{\lambda \in (0, \lambda_0]} \big[ - (n - \mu) \lambda  + \Lambda( \lambda ) \big]} \;  . \label{eq:lde}
\end{align}

\end{corollary}

We remark that Corollary \ref{cor:lde} provides through \eqref{eq:lde} an upper bound for the upper tails of $\mathcal{N}_+$. Related to \eqref{eq:lln}, we can, however, not prove a matching lower bound for the upper tails. 

Corollary \ref{cor:lde} illustrates correlations among excitations of the BEC. To be more precise, Corollary \ref{cor:lde} characterizes the behavior of the sum  $ \mathcal{N}_+ =  \sum_{i=1}^N Q_i$, where, in abuse of notation, we call $Q_{i}$ the random variables with law 
\begin{align}
\mathbb{P}\big[ Q_i \in B \big] = \langle \psi_N, \mathds{1}_{B} (Q_i ) \psi_N \rangle \quad \text{for any} \quad B \subset \mathbb{R} \; . 
\end{align}
The random variables are identically distributed (as $\psi_N \in L_s^2( \pi)$) and, due to quantum correlation, dependent.

\subsubsection*{Quadratic exponential bounds} Expanding the r.h.s. of \eqref{eq:lde} for small $\lambda >0$, we find that 
\begin{align}
\limsup_{N \rightarrow \infty} \mathbb{P}_{\psi_N} \big[ \mathcal{N}_+ - \mu \geq n \big]  \leq e^{\inf_{\lambda \in (0,\lambda_0]} \big[ -n \lambda + \frac{\lambda^2 }{2} \sigma^2 + \mathcal{O}( \lambda^3) \big]} \label{eq:quad}
\end{align}
with $\sigma^2>0$ given by \eqref{def:sigma}, and we recover back quadratic bounds for the upper tails of  $\mathcal{N}_+$ proven earlier in \cite{NR} through an expansion of the generating function $\mathbb{E}_{\psi_N} \big[ e^{\lambda \mathcal{N}_+} \big] $ for small $\lambda >0$ together with the exponential bounds \eqref{eq:expbounds}.

\subsubsection*{General one-particle observables} We further generalize Theorem \ref{thm:mgf} resp. Corollary \ref{cor:lde} to self-adjoint one-particle operators\footnote{Here we introduced the notation ${O}$ as the inverse Fourier transform of the operator $O$ defined in $\ell^2( \pi^*) \times \ell^2( \pi^*)$ since in our setting it is more convenient to formulate the following statement in momentum space $\pi^*$. If $O$ has kernel $O_{p,q} \in \ell^2( \pi^*) \times \ell^2(\pi^*)$, then the operator $\check{O}$ has kernel $\check{O}(x,y) = \sum_{p,q \in \pi^*} O_{p,q} e^{ip\cdot x} e^{-iq \cdot y}$. } $O$ on $\ell^2( \pi^*)$ satisfying $\check{O}= Q \check{O}Q$ for $Q = 1- \vert \varphi \rangle \langle \varphi \vert$, i.e. to one-particle operators $O$ that are orthogonal to the condensate.   We then define the $N$-particle operator 
\begin{align}
O_i = \mathds{1} \otimes \cdots \otimes \mathds{1} \otimes O \otimes \mathds{1} \otimes \cdots \otimes \mathds{1} 
\end{align}
acting as identity on all but the $i$-th particle on which it acts as the operator $O$. With this notation we define, in abuse of notation, the random variables $O_i$ through its law 
\begin{align}
\label{def:law}
\mathbb{P}\big[ O_i \in B \big] = \langle \psi_N, \mathds{1}_{B} (O_i ) \psi_N \rangle, \quad \text{for any} \quad B \subset \mathbb{R} \; . 
\end{align}
The next theorem proves an asymptotic upper bound for an exponential decay rate of the sum of the random variables $O_i$. To state our result we introduce the short-hand notation 
\begin{align}
s_p := \sinh( \nu_p), \quad c_p = \cosh( \nu_p) \;,\quad \text{and} \quad \tau_p = \tanh( \nu_p) 
\end{align}
that we use to formulate the asymptotic expectation value of $O_i$ 
\begin{align}
\label{def:muO}
\mu_O := \sum_{p \in \pi_+^*} O_{p,p} s_p^2 \; . 
\end{align}

\begin{theorem}\label{thm:lde-O} Let $V \in L^3( \pi)$ be non-negative, compactly supported and spherically symmetric and let $\psi_N$ denote the ground state of $H_N$ defined in \eqref{def:ham}. 

Furthermore, let $O$ be a self-adjoint operator on $\ell^2( \pi^*)$ with kernel $O_{p,q} \in \ell^2( \pi^*_+) \times \ell^2( \pi^*_+)$ satisfying $\check{O} = Q\check{O}Q$ where $Q$ is defined in \eqref{def:Q}. Then, there exists $\lambda_0 > 0$ such that the random variables $O_i$ for $i \in \lbrace 1, \dots, \rbrace$ defined by  \eqref{def:law} satisfy 
\begin{align}
\lim_{N \rightarrow \infty} \mathbb{E}_{\psi_N}\big[ e^{\lambda  \sum_{i=1}^N O_i } \big]  =e^{ \Lambda_O ( \lambda) }
\end{align}
for all $\vert \lambda \vert \leq \lambda_0 $ and, moreover, 
\begin{align}
\lim_{N \rightarrow \infty} \mathbb{P}_{\psi_N} \bigg[ \sum_{i=1}^N O_i - \mu_O  \geq n \bigg] \leq  e^{\inf_{0 < \lambda < \lambda_0} \big[ (n- \mu_0) \lambda - \Lambda_O ( \lambda ) \big] } \label{eq:lde-O}
\end{align}
for all $n \in (0,N ]$ where $\mu_O$ is defined in \eqref{def:muO} and $\Lambda : [0, \lambda_0] \rightarrow \mathbb{R}$ denotes a convex function given by 
\begin{align}
\Lambda ( \lambda ) =  \int_0^\lambda  \sum_{p,q,k,\ell \in \pi_+^*} s_p c_q O_{p,q} \sum_{j=0}^\infty \big( D_{p,q,k,\ell} ( \kappa) \big)^j A_{k, \ell} ( \kappa ) d \kappa + \lambda \sum_{p \in \pi_+^*} O_{p,p} s_p^2  \label{def:Lambda-O}
\end{align}
where the $j$-th power of the kernel $D_{p,q,k,\ell}$ is recursively defined by 
\begin{align}
D_{p,q,k,\ell}^0( \kappa) = \delta_{p,k} \delta_{q,\ell}, \quad \text{and} \quad D_{p,q,k,\ell}^j ( \kappa) =\sum_{k',\ell' \in \pi_+^*} D_{p,q,k',\ell'}^{j-1} ( \kappa) D_{k',\ell',k,\ell} ( \kappa) 
\end{align}
and 
\begin{align}
A_{p,q}( \kappa) =& - c_p s_q \delta_{p,q} + c_pc_q  \langle c e^{\kappa O_{-p}}, s  e^{\kappa \overline{O}^-_{q}} \rangle_{\ell^2 ( \pi_+^*)} +  s_p s_q \langle s e^{-\kappa \overline{O}_{-p}^-}, c e^{-\kappa O_q} \rangle_{\ell^2( \pi_+^*)} \notag \\
&- s_pc_q \overline{\langle s e^{-\kappa O_{-p}}, s e^{\kappa O_{-q}} \rangle_{\ell^2( \pi_+^*)}} - c_ps_q \overline{\langle s e^{-\kappa O_{q}}, s e^{\kappa O_{p}} \rangle_{\ell^2( \pi_+^*)}}  \label{def:Apq}
\end{align}
with the notation $O_p := O_{\cdot, p}, O^{-}_q = O_{- \cdot, q}$ and 
\begin{align}
\label{def:D}
 & D_{p,q, k,\ell} ( \kappa) \notag \\
 & \quad =c_p c_k  c_q  s_\ell \tau_\ell \bigg(  e^{\kappa \overline{O}_{-k,-p}}    e^{\kappa \overline{O}_{\ell,q}} + e^{-\kappa O_{k,-p}} e^{-\kappa O_{-\ell,q}} - e^{-\kappa O_{k,-p}}  e^{\kappa \overline{O}_{\ell, -q}}  - e^{-\kappa O_{k,p}} e^{\kappa \overline{O}_{\ell, q}}  \bigg)  \; . 
\end{align}
\end{theorem}

\begin{remark}
We collect some remarks on Theorem \ref{thm:lde-O}.
\begin{enumerate}
\item[(i)] In the proof of the Theorem in Section \ref{sec:proofs} we establish the rate of convergence of \eqref{eq:mgf-O} that is $\mathcal{O}( N^{-1/4})$ (see \eqref{eq:rate-O}). 
\item[(ii)] Furthermore, in the proof we show that the asymptotic generating function $\Lambda_O$ is well defined, i.e. that the r.h.s. of \eqref{def:Lambda-O} is finite (see the discussion before \eqref{eq:final-G-O}). 
\end{enumerate}
\end{remark}

\subsubsection*{Normal vs. anomalous random variables} Corollary \ref{cor:lde} and Theorem \ref{thm:lde-O} characterize tails for the sum of correlated random variables $O^i$ resp. $Q^i$ with non-vanishing deviation in the large particle limit; they show anomalous scaling properties. These results embed in a series of recent results \cite{BKS,RS,Rsing,KRS,RSe,R,PRV,COS} on a new probabilistic interpretation of Bose-Einstein condensates studying random variables   $A_i$ defined similarly to \eqref{def:law} with $A\not= QAQ$. These random variables $A_i$, thus, measure the correlation between the condensate and the excitations. Similarly to the random variables $Q_i$, the random variables $A_i$ are identically distributed and dependent. In contrast to \eqref{def:sigma}, the random variables cluster around their asymptotic mean
\begin{align}
\mathbb{E}_{\psi_N} \bigg[ \frac{1}{N} \sum_{i=1}^N A_i \bigg] \xrightarrow[N \rightarrow \infty]{} \langle \varphi, A \varphi \rangle 
\end{align}
with a deviation 
\begin{align}
\sqrt{\mathbb{E}_{\psi_N} \bigg[ \bigg(\frac{1}{N}\sum_{i=1}^N A_i - \langle \varphi, A \varphi  \rangle \bigg)^2 \bigg]} \xrightarrow[N \rightarrow \infty]{}  0   \label{eq:sd-A}
\end{align}
that  vanishes in the large particle limit. Consequently, Bose-Einstein condensation is associated with a weak law of large numbers of the random variables $A_i$ (see \cite{BKS}). In fact, the deviation \eqref{eq:sd-A} of the random variable $A_i$ from its mean turns out to be $\mathcal{O}(N^{-1})$, and thus agrees with the scaling of normal distributed random variables. In past years, the random variables were indeed proven to satisfy central limit theorems \cite{RS}, also recently in the dynamical setting \cite{COS}. In the mean-field regime (corresponding to the case $V_N := V$ in \eqref{def:ham}) the large deviation regime for $A_i$ is characterized through quadratic upper and (matching) lower bounds both for the ground state \cite{R} and the dynamics \cite{KRS,RSe}. Note that, for the random variables $O_i$ resp. $Q_i$, we can prove an upper bounds in Corollary \ref{cor:lde}, \ref{thm:lde-O} on their tails only.  

\subsection{Idea of the proof} The novelty this paper is the derivation of an explicit asymptotic formula for the generating function for the number of excitations $\mathcal{N}_+$ in Theorem \ref{thm:mgf} resp. the sum of random variables $\sum_{i=1}^N O_i$ (as defined in \ref{def:law}) in Theorem \ref{thm:lde-O}. The idea for establishing explicit formulas for 
\begin{align}
\label{eq:mgf-1}
\mathbb{E}_{\psi_N} \big[ e^{\lambda \mathcal{N}_+} \big] , \quad \text{resp.} \quad \mathbb{E}_{\psi_N} \big[ e^{\lambda \sum_{i=1}^N O_i} \big]
\end{align}
in the large particle limit $N \rightarrow \infty$, is to combine an approximation of the ground state $\psi_N$ in $L^2 (\pi^N)$-norm (proven in \cite{BCCS}) with exponential bounds on $\mathcal{N}_+$ recently proven in \cite{NR}. 
To be more precise, we show that thanks to \cite{BCCS} and \cite{NR},  the computation of \eqref{eq:mgf-1} reduces in the large particle limit $N \rightarrow \infty$ to the computation of the generating function of $\mathcal{N}_+$ in a quasi-free state, namely 
\begin{align}
\label{eq:mgf-4}
\mathbb{E}_{\psi_N} \big[ e^{\lambda \mathcal{N}_+} \big] \approx \langle e^{K_\nu} \Omega ,  e^{\lambda \mathcal{N}_+ }   e^{K_\nu} \Omega \rangle \; 
\end{align}
where $\Omega$ denote the vacuum vector of the bosonic Fock space, equipped with standard creation and annihilation operators $a^*_p, a_p$, and $e^{K_\nu}$ denotes a Bogoliubov transformation given by 
\begin{align}
\label{eq:bogo-null}
e^{K_\nu}, \quad \text{with} \quad K_\nu = \sum_{p \in \pi_+^*} \nu_p \big( a^*_{-p} a_p - a_p a_{-p} \big) 
\end{align}
with $\nu_p$ defined in \eqref{def:nu}. A peculiar property of a Bogoliubov transformation of the form \eqref{eq:bogo-null} is that one has explicit formulas for its action on creation and annihilation operators (for more details see Section \ref{sec:prel}) that, in particular lead to the validity of Wick's rule for quasi-free states (see for example \cite{Simon} and  \cite[Theorem 10.2]{SolovejLN}). Based on these, we explicitly compute the right-hand side of \eqref{eq:mgf-4} by solving an ordinary differential equation. To the best of the author's knowledge, this yields a novel formula for the moment generating function of $\mathcal{N}_+$ not previously reported in the literature. The computation of the asymptotic formula for the generating function of $\mathbb{E}_{\psi_N} \big[ e^{\lambda \sum_{i=1}^N O_i} \big]$ works with similar ideas. Both proofs are given in Section \ref{sec:proofs}.

\subsection{Structure of the paper} The rest of this paper is structured as follows: In Section \ref{sec:lln-lde} we use Theorem \ref{thm:mgf} to prove Corollary \ref{cor:lde} on the tails of the distribution of $\mathcal{N}_+$ and formulas \eqref{def:sigma}, \eqref{eq:sd-A} on the deviations from their mean. In Section \ref{sec:prel} we collect preliminary results, in particular on the approximation of the ground state and prove additional properties of its approximating state that we need to prove Theorems \ref{thm:mgf} and \ref{thm:lde-O} in Section \ref{sec:proofs}.

\section{Characterization of deviations and tails}

\label{sec:lln-lde}
In this Section we derive Corollary \ref{cor:lde} on the characterization of the tails (Section \ref{sec:ldp}) for $\mathcal{N}_+$ from the asymptotic generating function in Theorem \ref{thm:mgf}. We furthermore show that Theorem \ref{thm:mgf} allows to draw conclusion on the deviations of $\mathcal{N}_+$ from its mean $\mu$ (Section \ref{sec:corr}). 

These proofs are based on explicit asymptotic formulas of moments of the number of excitations. In fact from Theorem \ref{thm:mgf} (more specifically \eqref{eq:Nhochk}) we find 
\begin{align}
\label{eq:Nhochzwei}
\lim_{N \rightarrow \infty} \mathbb{E}_{\psi_N} \big[ (\mathcal{N}_+ - \mu)^2 \big] = \sigma^2
\end{align}
recovering back \eqref{def:sigma}. We recall definition 
\begin{align}
\sigma^2 = 2 \sum_{p \in \pi_+^*} \cosh^2( \nu_p)\sinh^2(\nu_p)
\end{align}
and the sequence $\nu_p$ is given by \eqref{def:nu}. A straight forward computation, based on \ref{eq:Nhochk} of Theorem \ref{thm:mgf}, shows furthermore 
\begin{align}
\label{eq:Nhochvier}
\lim_{N \rightarrow \infty} \mathbb{E}_{\psi_N} \big[ (\mathcal{N}_+ - \mu)^4 \big] = 12 \sigma^4 + 8 \sigma^2 + 48  \sum_{p \in \pi_+^*} \cosh( \nu_p )^4 \sinh( \nu_p)^4  \; . 
\end{align}

\subsection{Characterization of tails}
\label{sec:ldp}
 In this Section we prove Corollary \ref{cor:lde}.  

\begin{proof}[Proof of Theorem \ref{cor:lde}] The proof of both statements of Corollary \ref{cor:lde}, namely \eqref{eq:lde} and \eqref{eq:lln}, are both based on Theorem \ref{thm:mgf}. We start with the proof of \eqref{eq:lde} that is an immediate consequence of Markov's inequality. 

\subsubsection*{\textbf{Proof of \eqref{eq:lde}}:} For all $\lambda >0$ we have 
\begin{align}
\mathbb{P}_{\psi_N}\big[  \mathcal{N}_+ - \mu > n  \big] =\mathbb{P}\big[ e^{\lambda \mathcal{N}_+  - \lambda \mu} > e^{\lambda n}  \big] 
\end{align}
and, thus we find with Markov's inequality 
\begin{align}
\mathbb{P}_{\psi_N}\big[ \mathcal{N}_+ - \mu > n \big] \leq e^{-\lambda(n + \mu)} \mathbb{E} \big[ e^{\lambda \mathcal{N}_+} \big]  
\end{align}
Since this holds true for any $\lambda>0$, we conclude with Theorem \ref{thm:mgf} that 
\begin{align}
\limsup_{N \rightarrow \infty} \log \mathbb{P}_{\psi_N} \big[ \mathcal{N}_+ - \mu > n \big] \leq \inf_{\lambda \in (0,\lambda_0]} \big[ - \lambda (n + \mu)  + \Lambda (\lambda )\big]  \; .  
\end{align}

\subsubsection*{\textbf{Proof of \eqref{eq:lln}:}} The proof of \eqref{eq:lln} uses that the asymptotic generating function $\Lambda$ defined in \eqref{def:Lambda} allows to explicitly compute any moment of $\mathcal{N}_+$ asymptotically.

First we observe that for arbitrary $n,m>0$, that we will choose later, 
\begin{align}
\mathbb{P}_{\psi_N} \big[ \big\vert  \mathcal{N}_+ - \mu \big\vert > n \big] \geq \mathbb{P}_{\psi_N} \big[\vert \mathcal{N}_+ - \mu \vert \in (n,n+m)\big] \; . 
\end{align}
We proceed with estimating the r.h.s. from above. To this end, we introduce the notation $\widetilde{\mathcal{N}}_+ = \vert \mathcal{N}_+ - \mu \vert $  
\begin{align}
\mathbb{P}_{\psi_N} \big[ \vert \mathcal{N}_+ - \mu \vert \in (n,n+m)\big] =& \mathbb{E}_{\psi_N} \big[ \mathds{1}_{(n,n+m)} \big( \widetilde{\mathcal{N}}_+ \big) \big] \notag \\
=&  \mathbb{E}_{\psi_N} \big[  \mathds{1}_{(n,n+m)} \big( \widetilde{\mathcal{N}}_+ \big) \widetilde{\mathcal{N}}_+^2 \widetilde{\mathcal{N}}_+^{-2} \big]  \; . 
\end{align}
Since $\widetilde{\mathcal{N}}_+^{-2} \geq (n+m)^{-2}$ on the support of $\mathds{1}_{(n,n+m)} \big( \widetilde{\mathcal{N}}_+ \big)$, we find 
\begin{align}
\mathbb{P}_{\psi_N} \bigg[\vert \mathcal{N}_+ - \mu \vert \in (n,n+m)\bigg] \geq& \frac{1}{(n+m)^2} \mathbb{E}_{\psi_N} \big[ \mathds{1}_{(n,n+m)} \big(\widetilde{\mathcal{N}}_+ \big) \widetilde{\mathcal{N}}_+^{2} \big] \; . \label{eq:lb1}
\end{align}
We define the new probability distribution 
\begin{align}
\label{def:Ptilde}
\widetilde{\mathbb{P}}_{\psi_N} \big[ \widetilde{\mathcal{N}}_+ \in A \big] =\frac{\mathbb{E}_{\psi_N} \big[ \widetilde{\mathcal{N}}_+^{2} \chi_{A} ( \widetilde{\mathcal{N}}_+) \big]  }{\mathbb{E}_{\psi_N} \big[ \widetilde{\mathcal{N}}_+^{2} \big] } 
\end{align}
that is for sufficiently large $N$ well defined as $\mathbb{E}_{\psi_N} \big[ \widetilde{\mathcal{N}}_+^{2} \big] >0$ from \eqref{eq:Nhochzwei}. Thus with \eqref{def:Ptilde} we can write the r.h.s. of \eqref{eq:lb1} as 
\begin{align}
\mathbb{P}_{\psi_N} \big[ \vert \mathcal{N}_+ - \mu \vert \in (n,n+m)\big] \geq& \frac{\mathbb{E}_{\psi_N} \big[ \widetilde{\mathcal{N}}_+^{2} \big] }{(n+m)^2 } \widetilde{\mathbb{P}} \big[ \widetilde{\mathcal{N}}_+ \in (n,n+m)\big]  \notag \\
=& \frac{\mathbb{E}_{\psi_N} \big[ \widetilde{\mathcal{N}}_+^{2} \big] }{(n+m)^2 } \bigg( 1 - \widetilde{\mathbb{P}} \big[ \widetilde{\mathcal{N}}_+ \in (0,n] \big] - \widetilde{\mathbb{P}} \big[ \widetilde{\mathcal{N}}_+ \in [n+m, \infty) \big] \bigg)   \; . 
\end{align}
where we used that by definition \ref{def:Ptilde}, it follows $\widetilde{\mathbb{P}}_{\psi_N} \big[ \vert \mathcal{N}_+ - \mu \vert =0 \big] =0$. Next, we apply once more Markov's inequality for the probability measure $\widetilde{\mathbb{P}}$ and arrive at 
\begin{align}
\mathbb{P}_{\psi_N} \big[ \vert \mathcal{N}_+ - \mu \vert \in (n,n+m)\big] \geq& \frac{\mathbb{E} \big[ \widetilde{\mathcal{N}}_+^{2} \big] }{(n+m)^2 } \widetilde{\mathbb{P}} \big[ \widetilde{\mathcal{N}}_+ \in (n,n+m)\big]  \notag \\
=& \frac{\mathbb{E} \big[ \widetilde{\mathcal{N}}_+^{2} \big] }{(n+m)^2  } \bigg( 1 - \frac{1}{\mathbb{E} \big[ \widetilde{\mathcal{N}}_+^{2} \big]} \bigg( n^2 + \frac{1}{(n+m)^2} \mathbb{E} \big[ \widetilde{\mathcal{N}}^4 \big]  \bigg)  \bigg)  \; . 
\end{align}
Now we choose $n,m$ such that 
\begin{align}
n^2= \frac{1}{4} \mathbb{E} \big[ \widetilde{\mathcal{N}}_+^{2} \big], \quad \frac{1}{(n+m)^2} =\frac{1}{4} \frac{\mathbb{E} \big[ \widetilde{\mathcal{N}}_+^{2} \big]}{\mathbb{E} \big[ \widetilde{\mathcal{N}}_+^{4} \big]} 
\end{align}
yielding 
\begin{align}
\mathbb{P} \big[\vert  \mathcal{N}_+ - \mu \vert \in (n,n+m)\big] \geq& \frac{\mathbb{E} \big[ \widetilde{\mathcal{N}}_+^{2} \big] }{2 (n+m)^2 } = \frac{\mathbb{E} \big[ \widetilde{\mathcal{N}}_+^{2} \big]^2  }{ 8\mathbb{E} \big[ \widetilde{\mathcal{N}}_+^{4} \big] } \; . 
\end{align}
From \eqref{eq:Nhochvier} we find using the assumption $\a0 >0$, that in particular implies $\sigma >0$, that $\mathbb{E} \big[ \widetilde{\mathcal{N}}_+^{4} \big] \geq 12 \sigma^4$ so that we finally arrive with \eqref{eq:Nhochzwei} at 
\begin{align}
\mathbb{P} \big[ \vert \mathcal{N}_+ - \mu \vert \in (n,n+m)\big] \geq& \frac{\sigma^4}{96  \sigma^4} > 0  
\end{align}
that is the desired statement. 
\end{proof}

\subsection{Characterization of deviations from the mean} \label{sec:corr} In this section we show that the deviation of $\mathcal{N}_+$ from its mean $\mu$ does not vanish in the large particle limit (i.e. formula \eqref{def:sigma}) while, contrarily, the deviation of the random variables $A_i$ with $A \not= QAQ$ as defined in \eqref{def:law} does vanish as $N \rightarrow \infty$ (see formula  \eqref{eq:sd-A}). 

\subsubsection*{Proof of \eqref{def:sigma}} From \eqref{eq:Nhochzwei}, we immediately find 
\begin{align}
\lim_{N \rightarrow \infty} \sqrt{\mathbb{E}_{\psi_N} \big[ \big( \mathcal{N}_+ - \mu \big)^2 \big]} = \sigma \; . 
\end{align}

\subsubsection*{Proof of \eqref{eq:sd-A}} We compute 
\begin{align}
\label{eq:sd-A-1}
 \mathbb{E}_{\psi_N}&  \bigg[ \bigg( \frac{1}{N}\sum_{i=1}^N A_i - \langle \varphi, A \varphi \rangle \bigg)^2 \bigg] \notag \\
 &=  \frac{1}{N^2}  \mathbb{E}_{\psi_N} \bigg[ \sum_{i\not= j}^N A_i A_j \bigg] + \frac{1}{N^2}  \mathbb{E}_{\psi_N} \bigg[ \sum_{i=1}^N A_i \bigg]- \frac{2}{N} \langle \varphi, A \varphi \rangle \mathbb{E}_{\psi_N} \bigg[ \sum_{i=1}^N A_i \bigg] +  \langle \varphi, A \varphi \rangle^2
\end{align}
and use the $k$-particle reduced density associated to $\psi_N$
\begin{align}
\gamma_{\psi_N}^{(k)} = \Tr_{k+1, \dots, N} \vert \psi_N \rangle \langle \psi_N \vert 
\end{align}
defined as the partial trace starting from the $k+1$ to the $N$-th particle of the orthogonal projection onto $\psi_N$ to express the r.h.s. of \eqref{eq:sd-A-1} as 
\begin{align}
 \mathbb{E}_{\psi_N} &\bigg[ \bigg( \frac{1}{N}\sum_{i=1}^N A_i - \langle \varphi, A \varphi \rangle \bigg)^2 \bigg] \notag \\
  &=  \frac{N(N-1)}{N^2} \Tr \gamma_{\psi_N}^{(2)} ( A \otimes A ) + \frac{1}{N} \Tr \gamma_{\psi_N}^{(1)} A - 2\langle \varphi, A \varphi \rangle \; \Tr \gamma_{\psi_N}^{(1)} A + \langle \varphi, A \varphi \rangle^2 . 
\end{align}
Since, as proven in for example in \cite{LS-02}, for fixed $k \in \mathbb{N}$, we have 
\begin{align}
\lim_{N \rightarrow \infty} \Tr \big\vert \gamma_{\psi_N}^{(k)} - \vert \varphi \rangle \langle \varphi\vert^{\otimes k} \big\vert =0 
\end{align}
we get 
\begin{align}
\lim_{N \rightarrow \infty} \mathbb{E}_{\psi_N} \bigg[ \bigg( \frac{1}{N}\sum_{i=1}^N A_i - \langle \varphi, A \varphi \rangle \bigg)^2 \bigg]  =0 \; . 
\end{align}

\section{Preliminaries}
\label{sec:prel}

In this section, we recall basic notation and results from the literature on the many-body problem that we use for the proof of Theorem \ref{thm:mgf} in Section \ref{sec:proofs} later. For more details on Bogoliubov theory in the context of Bose gases, see for example the introductory lecture notes \cite{SolovejLN} or recent review articles \cite{Nap18, N21, S22}. 

\subsection{Approximation of the ground state}

For our analysis it is convenient to embed the $N$-body Hamiltonian \eqref{def:ham} into the bosonic Fock space given by
\begin{align}
\mathcal{F} = \bigoplus_{k=1}^\infty L_s^2 ( \pi^k ) \; .
\end{align}
The bosonic Fock space $\mathcal{F}$ is equipped with standard creation and annihilation operators $a^*(f)$ resp. $a(g)$ with $f,g \in L^2( \pi)$ that satisfy the canonical commutation relations 
\begin{align}
\label{eq:CCR}
\big[ a(g), a^*(f) \big] = \langle g,f \rangle , \quad \big[ a^*(f), a^*(g) \big] = \big[ a(f), a(g) \big] = 0 \; . 
\end{align}
For our analysis in the translation invariant setting on $\pi$, it is convenient to work in momentum space $\pi^* = ( 2 \pi \mathbb{Z})^3$ where we define 
\begin{align}
a_p^* = a^*( \varphi_p ) , \quad a_p = a( \varphi_p ), \quad \varphi_p (x) = e^{i p \cdot x} \in L^2( \pi ) \; . 
\end{align}
In momentum space, the number of particles operator can then be written as
\begin{align}
\mathcal{N} = \sum_{p \in \pi^*}a_p^*a_p 
\end{align}
and decomposed, recalling that the condensate wave function $\varphi$ corresponds to the zero mode $p=0$, as 
\begin{align}
\mathcal{N} = \mathcal{N}_+ + a_0^*a_0, \quad \text{with} \quad \mathcal{N}_+ = \sum_{p \in \pi_+^*}a_p^*a_p 
\end{align}
the number of excitations and the notation $\pi_+^* = \pi^* \setminus \lbrace 0 \rbrace$. We recall some bounds of creation and annihilation operators in terms of the number of particles operator, that we use in the proof later. For $\psi \in \mathcal{F}$ and $f \in L^2( \pi)$ we have 
\begin{align}
\label{eq:a-bounds1}
\| a(f) \psi \| \leq \| f \|_{L^2} \|  \mathcal{N}^{1/2} \psi \| , \quad \| a^*(f) \psi \leq \| f \|_{L^2} \| ( \mathcal{N} + 1)^{1/2} \psi \|,
\end{align}
and for any operator $H$ on $\ell^2( \pi^*)$ with kernel $H_{p,q}$ and $\xi_1, \xi_2 \in \mathcal{F}$, we have 
\begin{align}
\label{eq:a-bounds2}
\big\vert \langle \xi_1, \sum_{p,q \in \pi^*} H_{p,q} a_p^*a_{q} \xi_2 \rangle \vert \leq \| H \|_{\rm op} \| \mathcal{N}^{1/2} \xi_1 \| \; \| \mathcal{N}^{1/2} \xi_2 \|, 
\end{align}
and furthermore, 
\begin{align}
\label{eq:a-bounds3}
\big\vert \langle \xi_1, \sum_{p,q \in \pi^*} H_{p,q} a_p^*a_{-q} \xi_2 \rangle \vert \leq \| H \|_{\ell^2( \pi \times \pi)} \| \mathcal{N}^{1/2} \xi_1 \| \; \| (\mathcal{N} + 1)^{1/2} \xi_2 \| \; . 
\end{align}
With this notation, the second quantization of the Hamiltonian \eqref{def:ham} in momentum space reads 
\begin{align}
H_N = \sum_{p \in \Lambda^*} + \frac{1}{N} \sum_{r,p,q \in \Lambda^*} \widehat{V} (r/N) a_{p+r}^*a_q^*a_pa_{q+r} 
\end{align}
where we used the convention $\widehat{V} (p) = \int_{\mathbb{R}^3}dx \;  e^{-ip \cdot x} V(x) $ for $p \in \mathbb{R}^3$.

To study the quantum fluctuations around the condensate's behavior, we factor out any contribution of the condensate $\varphi = 1_{\pi}$ through the unitary 
\begin{align}
\mathcal{U}_N : L^2_s( \pi^N) \rightarrow \mathcal{F}_{\perp \varphi}^{\leq N} :=  \bigoplus_{n=0}^N L^2_{\perp \varphi} ( \pi)^{\otimes_s^N} 
\end{align}
that maps $\psi_N$, which can be uniquely decomposed as 
\begin{align}
\psi_N := \sum_{n=0}^N \varphi^{\otimes (N-n)} \otimes_s \xi^{(n)}, \quad \text{with} \quad \xi^{(n)}\in L_{\perp \varphi}^2( \pi)^{\otimes_s n} \; , 
\end{align}
onto the excitation vector $\xi := \lbrace \xi^{(0)}, \dots, \xi^{(n)} \rbrace$, i.e. $\mathcal{U}_N \psi_N = \xi $. The unitary $\mathcal{U}_N$, first introduced in \cite{LNSS}, acts for $p,q \in \pi_+^*$ on products of creation and annihilation operators as 
\begin{align}
\mathcal{U}_N a_p^* a_q \mathcal{U}_N^* = a_p^*a_q, \quad \mathcal{U}_N a_0^*a_0 \mathcal{U}_N^* = N - \mathcal{N}_+ ,
\end{align} 
resp. 
\begin{align}
\mathcal{U}_N a_p^*a_0 \mathcal{U}_N^* = N^{1/2} b_p^*, \quad  \mathcal{U}_N a_0^*a_q \mathcal{U}_N^* = N^{1/2} b_q
\end{align}
where we introduced modified creation and annihilation operators $b_p^*,b_q$ given for $p,q\in \pi^*_+$ by 
\begin{align}
\label{def:b}
b_p^* := a_p^* \sqrt{1-\mathcal{N}_+/N} , \quad b_p:= \sqrt{1-\mathcal{N}_+ /N} a_q \; . 
\end{align}
Contrarily to standard creation and annihilation operators, the modified ones, first introduced in \cite{BS}, leave the truncated Fock space $\mathcal{F}_{\perp \varphi}^{\leq N}$ invariant, however, that comes with the price of modified commutation relations 
\begin{align}
\label{eq:modifiedCCR}
[b_p, b_q ] = \bigg( 1- \frac{\mathcal{N}_+}{N} \bigg) \delta_{p,q} - \frac{1}{N} a_q^*a_p, \quad [b_p,b_q] = [b_p^*b_q^*] =0 
\end{align}
having, compared to \eqref{eq:CCR}, a correction term of order $N^{-1}$. Similarly to be standard creation and annihilation operators in \eqref{eq:a-bounds1}-\eqref{eq:a-bounds3}, the modified ones are bounded w.r.t. powers of the number of excitations. We have for $\psi \in \mathcal{F}_{\perp \varphi}^{\leq N}$ and $f \in L^2( \pi)$ we have 
\begin{align}
\label{eq:b-bounds1}
\| b(f) \psi \| \leq C \| f \|_{L^2} \|  \mathcal{N}_+^{1/2} \psi \| , \quad \| b^*(f) \psi \leq C \| f \|_{L^2} \| ( \mathcal{N}_+ + 1)^{1/2} \psi \|,
\end{align}
for $C>0$ and for any operator $H$ on $\ell^2( \pi^*)$ with kernel $H_{p,q}$ and $\xi_1, \xi_2 \in \mathcal{F}_{\perp \varphi}^{\leq N} $, we have 
\begin{align}
\label{eq:b-bounds2}
\big\vert \langle \xi_1, \sum_{p,q \in \pi^*} H_{p,q} b_p^*b_{q} \xi_2 \rangle \vert \leq  C \| H \|_{\rm op} \| \mathcal{N}_+^{1/2} \xi_1 \| \; \| \mathcal{N}_+^{1/2} \xi_2 \|, 
\end{align}
and furthermore, 
\begin{align}
\label{eq:b-bounds3}
\big\vert \langle \xi_1, \sum_{p,q \in \pi^*} H_{p,q} b_p^*b_{-q}^* \xi_2 \rangle \vert \leq  C \| H \|_{\ell^2( \pi^*_+ \times \pi_+^*)} \| \mathcal{N}_+^{1/2} \xi_1 \| \; \| (\mathcal{N}_+ + 1)^{1/2} \xi_2 \| \; . 
\end{align}
We remark that on the truncated Fock space, the number of excitations equals the number of particles operator, i.e. $\mathcal{N}_+ \psi = \mathcal{N} \psi$ for any $\psi \in \mathcal{F}_{\perp \varphi}^{\leq N}$. 

To study the quantum fluctuations, the so-called excitation Hamiltonian $\mathcal{L}_N = \mathcal{U}_N H \mathcal{U}_N^*$ needs to be regularized, through a modified Bogoliubov transformation that is a unitary map on $\mathcal{F}_{\perp \varphi}^{\leq N}$ given by 
\begin{align}
\label{def:bogo-mod1}
e^{B_\eta} = \exp \bigg[ \frac{1}{2} \sum_{p \in \Lambda_+^*} \eta_p \big( b_p^*b_{-p}^* - b_p b_{-p} \big) \bigg] \; . 
\end{align}
The choice of the sequence $\eta \in \ell^2( \pi_+^*)$ is such that it implements the particle's correlation structure and formulated in terms of the solution $f$ to the scattering equation with Neumann boundary conditions. To be more precise, for fixed, sufficiently large $\ell >0$, let $f_\ell$ denote the ground state solution to the Neumann problem 
\begin{align}
\bigg[ - \Delta + \frac{1}{2} V \bigg] f_\ell = \mu_\ell f_\ell 
\end{align}
on the ball $\vert x \vert \leq N \ell$, where the normalization is such that $f_\ell (x) =1$ of $\vert x \vert  = N \ell$ and $f_\ell (x) =1$ for all $\vert x \vert > N \ell$. Then, we set 
\begin{align}
\label{def:eta}
\check{\eta} : \pi \rightarrow L^2( \pi), \quad \check{\eta} (x) = - N \omega_\ell ( Nx)  
\end{align}
for the definition of the regularizing modified Bogoliubov transformation \eqref{def:bogo-modgen}. We remark that the sequence $\eta$ by definition \eqref{def:eta} satisfies (see for example \cite[Lemma 3.1]{BCCS}) 
\begin{align}
\label{eq:bounds-eta}
\vert \eta_p \vert \leq C \vert p \vert^{-2}  \; . 
\end{align}
and thus, in particular $\eta$ is an element of $\ell^2( \pi^*_+ )$. 

To prove the $L_s^2( \pi^N)$-norm approximation of the ground state in \cite{BCCS}, the Hamiltonian is further regularized using a unitary map on $\mathcal{F}_{\perp \varphi}^{\leq N}$ that is exponentially cubic in modified creation and annihilation operators and given by 
\begin{align}
\label{def:cubic}
e^{A_\eta}, \quad A_\eta = \exp \bigg[ N^{-1/2} \sum_{r \in P_H, v \in P_L} \eta_r \big( \sinh (\eta_v) b^*_{r+v} b_{-r}^* b_{-v}^* + \cosh_{\eta_v} b^*_{r+v} b^*_{-r}b_v - {\rm h.c.} \big) \bigg] 
\end{align}
with $\eta$ defined through \eqref{def:eta} and 
\begin{align}
P_L = \lbrace p \in \pi_+^* : \; \vert p \vert \leq N^{1/2} \rbrace \quad \text{and} \quad P_H = \pi_+^* \setminus P_L 
\end{align}
and a second modified Bogoliubov transform on $\mathcal{F}_{\perp \varphi}^{\leq N}$ given by 
\begin{align}
\label{def:bogo-mod2}
e^{B_\tau} = \exp \bigg[ \frac{1}{2} \sum_{p \in \pi_+^*} \tau_p \big( b_p^*b_{-p}^* - b_p b_{-p} \big) \bigg] 
\end{align}
where the sequence $\tau$ is given by 
\begin{align}
\label{def:tau}
\tau_p = - \frac{1}{4} \log \big[ 1 + 2 \vert p \vert^{-2} \widehat{V f_\ell} (p/N) \big] - \eta_p \; , 
\end{align}
and thus, in particular (see for example \cite[Lemma 5.1]{BCCS}) $\vert \tau_p \vert \leq C \vert p \vert^{-4}$ for all $p \in \pi_+^*$. It was proven in \cite[ Eq. (6.7)]{BCCS} that the ground state $\psi_N$ of $H_N$ then satisfies 
\begin{align}
\label{eq:norm}
\| \psi_N - e^{i \omega} e^{-B_\eta}e^{-A} e^{-B_\tau} \Omega \| \leq CN^{-1/4} 
\end{align}
for an appropriate choice $\omega \in [0,2\pi]$. While it is well known that the number of excitations is bounded in the ground state, recently it was proven \cite[Theorem 1]{NR} that also the exponential of the number of particles is bounded, i.e. that there exists $C>0$ such that 
\begin{align}
\label{eq:exp-bound}
\langle \psi_N, e^{\lambda \mathcal{N}_+} \psi_N \rangle \leq C 
\end{align}
for sufficiently small $\vert \lambda \vert >0$ (depending on $V$ only). 

\subsection{Modified Bogoliubov transformation}

In this section we study properties of general Bogoliubov transformations defined for any $\alpha \in \ell^2( \pi^*_+)$ by 
\begin{align}
\label{def:bogo-modgen}
e^{B_\alpha}, \quad \text{with} \quad B_\alpha = \frac{1}{2} \sum_{p \in \pi_+^*} \alpha_p \big( b_p^*b_{-p}^* - b_{p}b_{-p} \big) \;. 
\end{align}
In particular, the results then apply for $\eta$ defined in \eqref{def:eta} and $\tau$ defined in \eqref{def:tau} and thus for the two modified Bogoliubov transformation $e^{B_\tau}, e^{B_\eta}$ from \eqref{def:bogo-mod1},\eqref{def:bogo-mod2} that are relevant the for norm approximation of the ground state $\psi_N$ of \eqref{def:ham}. 

We remark that the action of the modified Bogoliubov transform on modified creation and annihilation operators is approximately known in the large particle limit. As proven in \cite[Lemma 2.3]{BCCS}, we have 
\begin{align}
\label{eq:bogo-mod-action}
e^{-B_\alpha} b_p e^{B_\alpha} =& \cosh ( \alpha_p) b_p + \sinh( \alpha_p) b_{-p}^* + d_p ,\notag \\
e^{-B_\alpha} b_p^* e^{B_\alpha} =& \cosh ( \alpha_p) b_p^* + \sinh( \alpha_p) b_{-p}+ d_p^* ,
\end{align}
where the error $d_p, d_p^*$ satisfy for any $\psi \in \mathcal{F}_{\perp \varphi}^{\leq N}$ and $n \in \mathbb{Z}$
\begin{align}
\label{eq:bounds-d}
\| ( \mathcal{N}_+ + 1)^{n/2} d_p^* \psi \| &\leq \frac{C_n}{N} \| ( \mathcal{N}_+ + 1)^{(n+3)/2} \psi \| \; , \notag \\
\| ( \mathcal{N}_+ + 1)^{n/2}  d_p \psi \| &\leq \frac{C_n}{N} \bigg( \vert \eta_p \vert \; \| ( \mathcal{N}_+ + 1)^{(n+3)/2} \psi \| + \| b_p ( \mathcal{N}_+ + 1)^{(n+2)/2} \psi \| 
\end{align}
for a constant $C_n >0$. 

Moreover, modified Bogoliubov transformations as in \eqref{def:bogo-modgen} are proven to approximately preserve powers of the number of excitations (see \cite[Lemma 2.1]{BCCS}), i.e. 
\begin{align}
\label{eq:bogo-mod-moment}
e^{-B_\alpha} ( \mathcal{N}_+ + 1)^k e^{B_\alpha} \leq C_k ( \mathcal{N}_+ + 1)^k \; . 
\end{align}
We improve that result and show that the modified Bogoliubov transform also approximately preserves the exponential of the number of excitations:  

\begin{lemma}
\label{lemma:bogo-mod}
For $\alpha \in \ell^2( \pi_+^*)$ the modified Bogoliubov transformation $e^{B_\alpha}$ defined in \eqref{def:bogo-mod1} satisfies for sufficiently small $\kappa>0$
\begin{align}
e^{-B_\alpha}  e^{\kappa \mathcal{N}_+} e^{B_\alpha} \leq e^{C \kappa (\mathcal{N}_+ + 1) }
\end{align}
for a constant $C>0$ as an operator inequality on $\mathcal{F}_{\perp \varphi}^{\leq N}$. 
\end{lemma}

\begin{proof}
The proof is based on a Gronwall argument. To this end we define for $s \in [0,1]$ and arbitrary $\psi \in \mathcal{F}_{\perp \varphi}^{\leq N}$ 
\begin{align}
\xi_s  =   e^{ c_s \kappa \mathcal{N}_+/2} e^{sB_\alpha} \psi 
\end{align}
for a monotone decreasing function $c_s :[0,1] \rightarrow \mathbb{R}_+$ with $c_1 = 1$ that we choose later. Then 
\begin{align}
\| \xi_1 \|^2 = \langle \psi, e^{-B_\alpha}  e^{\kappa \mathcal{N}_+} e^{B_\alpha} \psi \rangle, \quad \text{and} \quad \| \xi_0 \|^2  = \langle \psi,  e^{ c_0 \kappa \mathcal{N}_+}  \psi \rangle  \; .  
\end{align}
Thus in order to prove Lemma \ref{lemma:bogo-mod}, it is enough to control the derivative
\begin{align}
\partial_s \| \xi_s \|^2 =& 2 \Re \langle \xi_s, \partial_s \xi_s \rangle \notag \\
=& 2 \Re \langle \xi_s, \; \big( \dot{c}_s \kappa \mathcal{N}_+ + e^{c_s \kappa \mathcal{N}_+} B_\alpha e^{-c_s \kappa \mathcal{N}_+} \big) \xi_s \rangle 
\end{align}
and the idea is to control the second term of the r.h.s. through the first, choosing $c_s$ appropriately. Since 
\begin{align}
\label{eq:intertwining}
\mathcal{N}_+ b_p = b_p ( \mathcal{N}_+ - 1)  \quad \text{and} \quad \mathcal{N}_+ b_p^* = b_p^* ( \mathcal{N}_+ + 1)
\end{align}
we have 
\begin{align}
\Re \langle \xi_s ,   & \; e^{c_s \kappa \mathcal{N}_+} B_\alpha e^{-c_s \kappa  \mathcal{N}_+} \xi_s \rangle \notag \\
=& \Re  \sum_{p \in \pi_+^*} \alpha_p \; \langle \xi_s, \;  \big( e^{-2 \kappa c_s } b_p^*b^*_{-p} -  e^{2 \kappa c_s} b_pb_{-p} \big) \xi_s \rangle  \notag \\
=& \frac{1}{2} \Re \sum_{p \in \pi_+^*} \alpha_p \; \langle \xi_s, \; \big( \big( e^{-2 \kappa c_s } -e^{2\kappa c_s } \big)  b_p^*b^*_{-p} - \big(  e^{2 \kappa c_s} - e^{-2\kappa c_s }\big)  b_pb_{-p} \big) \xi_s \rangle \notag \\
=&  \sinh( 2 \kappa c_s) \sum_{p \in \pi_+^*} \alpha_p \langle \xi_s, \big( b_p^* b_{-p}^* + b_p b_{-p} \big) \xi_s \rangle \; . 
\end{align}
Now, using that $\alpha \in \ell^2( \pi_+^*)$, it thus follows for sufficiently small $c_s \kappa>0$ that $\sinh( 2\kappa c_s ) \leq \tilde{C} c_s \kappa $ for some $\tilde{C} >0$ and therefore 
\begin{align}
\label{eq:bound1}
\partial_s \| \xi_s \|^2 \leq  \kappa \langle \xi_s, \big( \dot{c}_s  + 2 \| \alpha \|_{\ell^2} \widetilde{C} c_s   \big) \mathcal{N}_+ \xi_s \rangle  + 2  \tilde{C} c_s \kappa \| \alpha \|_{\ell^2} \| \xi_s \|^2 \; .
\end{align}
Now we choose 
\begin{align}
c_s = e^{C \| \alpha\|_2 (1-s) } 
\end{align}
with $C>2\tilde{C}$. Then $c_1 = 1$ and $c_0 = e^{C \| \alpha \|_{\ell^2}}$ and $\dot{c}_s = - C \| \alpha \|_{\ell^2}  c_s $, and in particular $\dot{c}_s  + 2 \tilde{C} \| \alpha \|_{\ell^2}c_s < 0$ leading with \eqref{eq:bound1} to 
\begin{align}
\partial_s \| \xi_s \|^2  \leq 2 c_s \kappa \| \alpha \|_{\ell^2} \| \xi_s \|^2  \; 
\end{align}
that by Gronwall's inequality yields the desired bound of Lemma \ref{lemma:bogo-mod}. 
\end{proof}

\subsection{Properties of the cubic unitary $e^{A_\eta}$}

Here we discuss properties of the cubic unitary $e^{A_\eta}$ defined in \eqref{def:cubic}. As proven in \cite[Lemma 5.1]{BCCS}, the cubic unitary approximately preserves powers of the number of excitations, i.e. 
\begin{align}
\label{eq:cubic-moment}
e^{-A_\eta} ( \mathcal{N}_+ + 1)^k e^{A_\eta} \leq C_k ( \mathcal{N}_+ + 1)^k
\end{align}
for fixed $k \in \mathbb{N}$. Similarly as in the previous section we show that $e^{A_\eta}$ approximately preserves the exponential of the number of particles. 

\begin{lemma}
\label{lemma:cubic}
For $\eta \in \ell^2( \pi_+^*)$, the cubic transformation $e^{A_\eta}$ defined in \eqref{def:cubic} satisfies for sufficiently small $\kappa>0$
\begin{align}
e^{-A_\eta} e^{\kappa \mathcal{N}_+} e^{A_\eta} \leq e^{C \kappa ( \mathcal{N}_+ + 1)} 
\end{align}
for a constant $C>0$ as an operator inequality on $\mathcal{F}_{\perp \varphi}^{\leq N}$. 
\end{lemma}

\begin{proof}
Similarly as in the proof of Lemma \ref{lemma:bogo-mod}, we use a Gronwall type argument for the Fock space vector 
\begin{align}
\xi_s = e^{c_s \kappa \mathcal{N}_+ /2} e^{ s A_\eta} \psi, 
\end{align}
where we choose the monotone decreasing function $c_s:[0,1] \rightarrow \mathbb{R}_+$ with $c_1 = 1$ later. Since 
\begin{align}
\| \xi_1 \|^2 = \langle \psi, e^{-A_\eta} e^{\kappa \mathcal{N}_+} e^{A_\eta} \psi \rangle, \quad \text{and} \quad \| \xi_0 \|^2 = \langle \xi, e^{c_0 \kappa \mathcal{N}_+} \psi \rangle  \; . 
\end{align}
we need to control 
\begin{align}
\label{eq:deriv-xi}
\partial_s \| \xi_s \|^2 = 2 \Re \langle \xi_s, \big( \dot{c}_s \kappa
\mathcal{N}_+ + e^{c_s \kappa \mathcal{N}_+} A_\alpha e^{-c_s \kappa \mathcal{N}_+} \big) \xi_s \rangle \;  . 
\end{align}
With \eqref{eq:intertwining} we find 
\begin{align}
\Re & \; e^{c_s \kappa  \mathcal{N}_+} A_\alpha e^{-c_s \kappa \mathcal{N}_+} \notag \\
&= 2 N^{-1/2} \sum_{r \in P_H, v \in P_L} \eta_r \big(\sinh( 3c_s \kappa )  \sinh (\eta_v) b^*_{r+v} b_{-r}^* b_{-v}^* + \sinh( \kappa c_s)  \cosh (\eta_v ) b^*_{r+v} b^*_{-r}b_v  + {\rm h.c.} \big) \; . 
\end{align}
On the one hand we have for any $\psi \in \mathcal{F}_{\perp \varphi}^{\leq N}$
\begin{align}
\vert \langle \psi,  & \; \sum_{r \in P_H, v \in P_L} \eta_r  \sinh (\eta_v) b^*_{r+v} b_{-r}^* b_{-v}^* \psi \rangle \vert \notag \\
\leq& \bigg( \sum_{r \in P_H, v \in P_L} \| b_{r+v} b_{-r} \psi \|^2 \bigg)^{1/2} \bigg( \sum_{r \in P_H, v \in P_L} \vert \eta_r \vert^2 \vert \sinh( \eta_v)\vert^2  \| b_{-v} \psi \|^2 \bigg)^{1/2}
\end{align}
Since $\vert \eta_r \vert \leq C r^{-2}$ and $\vert \sinh( \eta_v) \vert^2 \leq \vert \eta_v \vert \leq C v^{-2}$ from \eqref{eq:bounds-eta}, it follows 
\begin{align}
\vert \langle \psi,  & \; \sum_{r \in P_H, v \in P_L} \eta_r  \sinh (\eta_v) b^*_{r+v} b_{-r}^* b_{-v}^* \psi \rangle \vert \leq C_1 \| \mathcal{N}_+ \psi \| \; \| ( \mathcal{N}_+ + 1)^{1/2} \psi \| \; .\label{eq:cubic-bound1}
\end{align}
for a constant $C_1>0$. On the other hand, we have 
\begin{align}
\vert \langle \psi, & \; \sum_{r \in P_H, v \in P_L} \eta_r  \cosh (\eta_v ) b^*_{r+v} b^*_{-r}b_v \psi \rangle \vert \notag \\
\leq& \bigg( \sum_{r \in P_H, v \in P_L} \| b_{r+v} b_{-r} \psi \|^2 \bigg)^{1/2} \bigg( \sum_{r \in P_H, v \in P_L} \vert \eta_r \vert^2 \vert \cosh( \eta_v)\vert^2 \| b_v \psi \|^2 \bigg)^{1/2}
\end{align}
and since $\vert \cosh(\eta_v) \vert  \leq C$ with \eqref{eq:bounds-eta} , it follows 
\begin{align}
\vert \langle \psi, & \; \sum_{r \in P_H, v \in P_L} \eta_r  \cosh (\eta_v ) b^*_{r+v} b^*_{-r}b_v \psi \rangle \vert \leq C_2 \| \mathcal{N}_+ \psi \| \; \| ( \mathcal{N}_+ + 1)^{1/2} \psi \| \; . \label{eq:cubic-bound2}
\end{align}
for some constant $C_2>0$. With \eqref{eq:cubic-bound1} and \eqref{eq:cubic-bound2}, we thus find for sufficiently small $\kappa c_s >0$ that 
\begin{align}
\vert \Re \langle \psi, e^{c_s\mathcal{N}_+} A_\alpha e^{-c_s \mathcal{N}_+}\psi \rangle \vert \leq C_3 c_s \kappa \langle \psi, ( \mathcal{N}_+ + 1) \psi \rangle 
\end{align}
for some $C_3>0$ that yields with the choice $c_s = e^{C_3 (1-s)}$ from \eqref{eq:deriv-xi} in 
\begin{align}
\partial_s \| \xi (s) \|^2 \leq \kappa \langle \xi(s), \big( - C_3 + C_3 ) \mathcal{N}_+ \xi(s) \rangle + C_3\kappa  \| \xi(s) \|^2 \leq C_3 \kappa \| \xi(s) \|^2 \; . 
\end{align}
Then Lemma \ref{lemma:cubic} follows with Gronwalls argument. 
\end{proof}

\subsection{Properties of the asymptotic Bogoliubov transformation}

In the proof of Theorem, we will show that the action of the three unitaries $e^{B_\eta} e^{A_\eta} e^{B_\tau}$ is asymptotically given through a (standard) Bogoliubov transformation w.r.t. standard creation and annihilation operators that is a unitary map on the full bosonic Fock space $\mathcal{F}$ and given for $\alpha \in \ell^2( \pi_+^*)$ by 
\begin{align}
\label{def:bogo}
e^{K_\alpha}, \quad \text{with} \quad K_\alpha = \sum_{p \in \pi_+^*} \alpha_p \big( a_{-p}^*a_p - a_p a_{-p} \big) , \quad \text{for} \quad \alpha \in \ell^2( \pi_+^* ) \; . 
\end{align}
In contrast to the modified Bogoliubov transform (see \eqref{eq:bogo-mod-action}), the action of the standard Bogoliubov transform on creation and annihilation operators is exactly known to be 
\begin{align}
\label{eq:bogo-action}
e^{-K_\alpha} a_p e^{K_\alpha} =& \cosh ( \alpha_p) a_p + \sinh( \alpha_p) a_{-p}^* \notag \\
e^{-K_\alpha} a_p^* e^{K_\alpha} =& \cosh ( \alpha_p) a_p^* + \sinh( \alpha_p) a_{-p} \; . 
\end{align}
We remark that as an immediate consequence of \eqref{eq:bogo-action} we have for any $h \in \ell^2( \pi^*)$ from 
\begin{align}
\label{eq:bogo-action-h}
e^{-K_\alpha} a(h) e^{K_\alpha} =& a \big( \overline{\cosh ( \alpha)} h\big)  + a^*\big( \sinh( \alpha) \overline{h^{-}} \big) \notag \\
e^{-K_\alpha} a_p^* e^{K_\alpha} =& a^* \big( \overline{\cosh ( \alpha)} h\big)  + a\big( \sinh( \alpha) \overline{h^{-}} \big)  \; 
\end{align}
where we used for any $h \in \ell^2( \pi_+^*)$ 
\begin{align}
\label{eq:not-minus}
h^{-}_p = h_{-p} \; . 
\end{align}

It is well known that the standard Bogoliubov transform approximately preserves powers of the number of excitations, i.e. 
\begin{align}
\label{eq:bogo-mom}
e^{-K_\alpha} ( \mathcal{N}_+ + 1)^k e^{K_\alpha} \leq C_k ( \mathcal{N}_+ + 1)^k \; . 
\end{align}
We prove that it approximately preserves the exponential of the number of excitations, similarly as the modified Bogoliubov transform (see Lemma \ref{lemma:bogo-mod}).

\begin{lemma}
\label{lemma:bogo-standard}
For $\alpha \in \ell^2( \pi^*_+)$, the standard Bogoliubov transform $e^{K_\alpha}$ defined in \eqref{def:bogo} satisfies for $\kappa>0$ 
\begin{align}
e^{-K_\alpha} e^{\kappa \mathcal{N}_+} e^{K_\alpha} \leq e^{ C \kappa( \mathcal{N}_+ + 1)}
\end{align}
for a constant $C>0$ as an operator inequality on $\mathcal{F}$. 
\end{lemma}

\begin{proof}
The proof of Lemma \ref{lemma:bogo-standard} follows with similar arguments as in the proof of Lemma \ref{lemma:bogo-mod}, using instead of the estimates \eqref{eq:b-bounds1} for the modified creation and annihilation operators, the similar estimates \eqref{eq:a-bounds1} for the standard ones. 
\end{proof}

\subsection{Baker-Campbell-Hausdorff formulas}

In this section we prove specific Baker-Campbell-Hausdorff formulas for non-commuting operators, that allow to explicitly compute the conjugation $e^{d \Gamma (O)} a_p e^{- d \Gamma (O)}$ for $p \in \pi^*$. For this we introduce for an operator $O$ on $\ell^2( \pi_+^*)$ the notation 
\begin{align}
\label{eq:not-O-p}
O_p := O_{\cdot,p}
\end{align}
in particular leading to $e^{O_p} = e^{O_{\cdot, p}}$.

\begin{lemma}
\label{lemma:BCH} 
Let $h \in \ell^2( \pi^*)$ and $O$ be a bounded operator on $\ell^2( \pi^*)$. Then, we have with the notation \eqref{eq:not-O-p} 
\begin{align}
e^{ d \Gamma (O)} a^*_p e^{- d \Gamma (O)} = a^*( e^{O_p} ), \quad \text{and} \quad e^{d \Gamma (O)} a_p e^{-d \Gamma (O)} = a( e^{-O_p} ) \;  
\end{align}
\end{lemma}

\begin{proof} The proof is based on the identity 
\begin{align}
\label{eq:series}
e^{ d \Gamma (O)} a^*_p e^{-d \Gamma (O)} = \sum_{j=0}^\infty \frac{1}{j!} {\rm ad}_{d \Gamma (O)}^{(j)} \big( a^*_p \big) 
\end{align}
where we introduced the notation of the $j$-th nested commutator, defined recursively by
\begin{align}
{\rm ad}_{d\Gamma (O)}^{(0)} ( a^*_p) = a^*_p, \quad {\rm ad}_{d\Gamma (O)}^{(j+1)} ( a^*_p) = \big[d\Gamma (O) , {\rm ad}_{d\Gamma (O)}^{(j)} ( a^*_p) \big] \; . 
\end{align}
We claim that the $j$-th nested commutator is given by 
\begin{align}
\label{eq:ind-hyp}
{\rm ad}_{d \Gamma (O)}^{(j)} \big( a^*_p \big)  = a^*( O_p^j)  
\end{align}
with the notation 
\begin{align}
O_p := O_{\cdot,p}, \quad \text{and} \quad O_p^{j} = O^{j}_{\cdot, p} \; . 
\end{align}
In fact \eqref{eq:ind-hyp} follows by induction: Assuming \eqref{eq:ind-hyp} to hold true for fixed $j \in \mathbb{N}$, we find 
\begin{align}
{\rm ad}_{d \Gamma (O)}^{(j+1)} \big( a^*(h) \big) =& \big[ d \Gamma (O), {\rm ad}_{d\Gamma (O)}^{(j)} ( a^*_p) \big] \notag \\
&= \sum_{k, \ell \in \pi^*} O_{k,\ell} O_{m,p}^{j} \big[ a_k^*a_\ell, a_m^* \big]  ) \notag \\
&= \sum_{k, \ell \in \pi^*} O_{k,p} O_{\ell,p} a^*k = a^*(O^{j+1}_p )   
\end{align}
and \eqref{eq:ind-hyp} follows for all $j \in \mathbb{N}$. 
Then the first equality of Lemma \ref{lemma:BCH} follows from \eqref{eq:ind-hyp} and \eqref{eq:series}. The second equality of Lemma \ref{lemma:BCH} then follows by taking the hermitian conjugate of the first and replacing $O$ with $-O$. 
\end{proof}

\section{Proof of Theorems} 
\label{sec:proofs}

The proof of Theorem \ref{thm:mgf} and Theorem \ref{thm:lde-O} use the same strategy thus we present them together. 

\subsection{Proof of Theorem \ref{thm:mgf} and Theorem \ref{thm:lde-O}} 

The proofs are split in three steps: 

We first show, based on the norm-approximation of the ground state and the exponential bounds \eqref{eq:exp-bound}, that the generating function of $d\Gamma (O)$ resp. $\mathcal{N}_+$ w.r.t. to the ground state $\psi_N$ is well approximated by 
\begin{align}
\label{eq:limit}
\langle \Omega, e^{-K_\nu} e^{\lambda d\Gamma (O)} e^{K_\nu} \Omega \rangle , \quad \text{resp.} \quad \langle \Omega, e^{-K_\nu} e^{\lambda \mathcal{N}_+} e^{K_\nu} \Omega \rangle 
\end{align}
for Theorem \ref{thm:mgf} resp. Theorem \ref{thm:lde-O} that is the generating function of $\mathcal{N}_+$ resp. $d \Gamma (O)$ w.r.t. to the quasi-free state $e^{K_\nu} \Omega$ where 
\begin{align}
\label{def:bogo-nu}
e^{K_\nu}, \quad \text{with} \quad K_\nu = \sum_{p \in \pi_+^*} \nu_p \big( a_{-p}^*a_p - a_p a_{-p} \big) , \quad \text{for} \quad \nu_p \quad \text{given by} \quad \eqref{def:nu} \; . 
\end{align}
This is summarized in the following Lemma:

\begin{lemma}
\label{lemma:approx}
Under the same assumption as in Theorem \ref{thm:mgf}, assume that $O$ is a bounded operator on $\ell^2( \pi^*_+)$. Then, for $k \in \mathbb{N}$, there exists $C_k>0$ such that 
\begin{align}
\big\vert \langle \psi_N,  \mathcal{N}_+^k  e^{\lambda d \Gamma (O)}  \psi_N \rangle -\langle \Omega, e^{-K_\nu}  \mathcal{N}_+^k  e^{\lambda d \Gamma (  O)} e^{K_\nu} \Omega \rangle \big\vert \leq C_k N^{-1/4}   \; . 
\end{align}
where $K_\nu$ is defined by \eqref{def:bogo-nu}. 
\end{lemma}

We remark that Lemma \ref{lemma:approx} in particular holds for the choice $O= \mathds{1}$ and thus for $d\Gamma ( \mathds{1}) = \mathcal{N}_+$. We prove Lemma \ref{lemma:approx} in Section \ref{sec:proof-lemma-approx} below and proceed with the proof of the theorems.

The second step of the proofs of Theorem \ref{thm:mgf} resp. Theorem \ref{thm:lde-O}, are the explicit computation of the generating functions \eqref{eq:limit} w.r.t to the quasi-free state $ e^{K_\nu} \Omega$ given by the following Lemma:

\begin{lemma}
\label{lemma:mfg}
Under the same assumptions as in Theorem \ref{thm:mgf} and Theorem \ref{thm:lde-O}, we have for sufficiently small $\vert  \lambda\vert \geq 0$  
\begin{align}
\langle \Omega, e^{-K_\nu} e^{\lambda \mathcal{N}_+} e^{K_\nu} \Omega \rangle  = e^{ \Lambda (\lambda)} \label{eq:mgf-N}
\end{align}
where $\Lambda$ is given by \eqref{def:Lambda} resp. 
\begin{align}
\langle \Omega, e^{-K_\nu} e^{\lambda d\Gamma (O)} e^{K_\nu} \Omega \rangle  =  e^{\Lambda_O ( \lambda )} \label{eq:mgf-O}
\end{align} 
where $\Lambda_O$ is defined in \eqref{def:Lambda-O}. 
\end{lemma}

Again we postpone the proof of Lemma \ref{lemma:mfg} to Section \ref{sec:proof-lemma-mgf} below and continue with the theorems' proofs. 

As a last, third step, it remains to combine Lemma \ref{lemma:approx} and Lemma \ref{lemma:mfg} to prove Theorem \ref{thm:mgf} and Theorem \ref{thm:lde-O}.

\begin{proof}[Proof of Theorem \ref{thm:mgf}] 
 The proof of \eqref{eq:thm-mgf} now follows immediately from Lemma \ref{lemma:approx} and Lemma \ref{lemma:mfg}. For sufficiently small $\lambda >0$, we find 
\begin{align}
\big\vert\mathbb{E}_{\psi_N} \big[ e^{\lambda \mathcal{N}_+} \big] - e^{\Lambda ( \lambda)} \big\vert \leq CN^{-1/4}   \; \; .\label{eq:mfg-approx}
\end{align} 
From Lemma \ref{lemma:mfg} we furthermore get for $\lambda$ with sufficiently small $\vert \lambda \vert$
\begin{align}
\Lambda ( \lambda) = \log e^{\Lambda ( \lambda)} = \log \langle \Omega, e^{-K_\nu} e^{\lambda \mathcal{N}_+ } e^{K_\nu} \Omega \rangle 
\end{align}
and since 
\begin{align}
\frac{d^2}{d\lambda^2}\Lambda ( \lambda) =&  \frac{d^2}{d\lambda^2} \log \langle \Omega, e^{-K_\nu}  e^{\lambda \mathcal{N}_+ } e^{K_\nu} \Omega \rangle  \notag \\
=& \frac{\langle \Omega, e^{-K_\nu}  \mathcal{N}_+^2 e^{\lambda \mathcal{N}_+ } e^{K_\nu} \Omega \rangle}{\langle \Omega, e^{-K_\nu}  e^{\lambda \mathcal{N}_+ } e^{K_\nu} \Omega \rangle} - \frac{\langle \Omega, e^{-K_\nu} \mathcal{N}_+ e^{\lambda \mathcal{N}_+ } e^{K_\nu} \Omega \rangle^2}{\langle \Omega, e^{-K_\nu}  e^{\lambda \mathcal{N}_+ } e^{K_\nu} \Omega \rangle^2} \geq 0
\end{align}
by Cauchy Schwarz, and thus, we conclude that $\Lambda$ is convex. 

Similarly, for the proof of \eqref{eq:Nhochk}, we find from Lemma \ref{lemma:approx} for sufficiently small $\lambda>0$ and fixed $k \in \mathbb{N}$  
\begin{align}
\big\vert \mathbb{E}_{\psi_N} \big[ \mathcal{N}_+^k e^{\lambda \mathcal{N}_+ } \big] - \langle \Omega,e^{-K_\nu}  \mathcal{N}_+^k e^{ \lambda \mathcal{N}_+} e^{K_\nu} \Omega \rangle  \big\vert \leq C _k N^{-1/4}  
\end{align}
and with Lemma \ref{lemma:mfg} 
\begin{align}
\langle \Omega,e^{-K_\nu}  \mathcal{N}_+^k e^{\lambda \mathcal{N}_+ } e^{K_\nu} \Omega \rangle = \frac{d^k}{d\lambda^k} \langle \Omega,e^{-K_\nu} e^{\lambda \mathcal{N}_+} e^{K_\nu} \Omega \rangle = \frac{d^k}{d\lambda^k}  \exp ( \Lambda( \lambda) ) 
\end{align}
we arrive at \eqref{eq:Nhochk}. 
\end{proof}

\begin{proof}[Proof of Theorem \ref{thm:lde-O}] 
The proof of \eqref{eq:mgf-O} follows from a combination of Lemma \ref{lemma:approx} and Lemma \ref{lemma:mfg} that show 
\begin{align}
\big\vert \mathbb{E}_{\psi_N} \big[ e^{\lambda \sum_{i=1}^N O_i } \big] - e^{ \Lambda_O ( \lambda)} \big\vert \leq CN^{-1/4} \label{eq:rate-O} 
\end{align}
for sufficiently small $\lambda >0$. Then \eqref{eq:lde-O} follows by Markov's inequality (similarly as in the proof of Corollary \ref{cor:lde} in Section \ref{sec:ldp}). 
\end{proof}

\subsection{Step 1: Approximation by quasi-free state}
\label{sec:proof-lemma-approx}

In this section we prove Lemma \ref{lemma:approx}.  

\begin{proof}[Proof of Lemma \ref{lemma:approx}] 
We split the proof in three steps: We first use show that, using the approximation \eqref{eq:norm}, we can replace $\psi_N$ with $e^{B_\tau} e^{A_\eta} e^{B_\eta} \Omega$. Second we show that the action of the cubic transform $e^{A_\eta}$ on $e^{\lambda d\Gamma (O)}$ is negligible in the large particle limit. In the last step, we then show that we can replace the action of the two modified Bogoliubov transformations $e^{B_\tau} e^{B_\eta}$ on $e^{\lambda d \Gamma (O)}$ by the action of an effective standard Bogoliubov transformation. 

\textbf{Step 1.1 (Norm approximation):} First we use the norm approximation \eqref{eq:norm} and find 
\begin{align}
\langle & \psi_N,  \mathcal{N}_+^k e^{\lambda d \Gamma (O)} \psi_N \rangle  -  \langle \Omega, e^{-K_\nu}\mathcal{N}_+^k  e^{\lambda d \Gamma (O)} e^{K_\nu} \Omega \rangle \notag \\
=& \langle\big(\psi_N - e^{i \omega} e^{-B_\eta}e^{A_\eta} e^{-B_\tau} \Omega \big) , \; \mathcal{N}_+^k e^{\lambda d \Gamma (O)} \psi_N \rangle \notag \\
&+ \langle e^{i \omega} e^{-B_\eta}e^{A_\eta} e^{-B_\tau} \Omega  , \; \mathcal{N}_+^k e^{\lambda d \Gamma (O)} \big(\psi_N - e^{i \omega} e^{-B_\eta}e^{A_\eta} e^{-B_\tau} \Omega \big) \rangle 
\end{align}
and thus, with \eqref{eq:norm} we find 
\begin{align}
\big\vert & \langle \psi_N, \mathcal{N}_+^k e^{\lambda d \Gamma (O)} \psi_N \rangle   -  \langle \Omega, e^{-K_\nu} \mathcal{N}_+^k e^{\lambda d\Gamma (O)} e^{K_\nu} \Omega \rangle\big\vert  \notag \\ 
\leq& \| \psi_N - e^{i \omega} e^{-B_\eta}e^{A_\eta} e^{-B_\tau} \Omega \| \; \bigg(  \|\mathcal{N}_+^k  e^{\lambda d\Gamma (O)} \psi_N \| +   \| \mathcal{N}_+^k e^{\lambda d \Gamma (O)} e^{-B_\eta}e^{A_\eta} e^{-B_\tau} \Omega \| \bigg)   \notag \\
\leq& C N^{-1/4} \bigg(  \|\mathcal{N}_+^k e^{\lambda d\Gamma (O)} \psi_N \| +   \| \mathcal{N}_+^ke^{\lambda d\Gamma (O)} e^{-B_\eta}e^{A_\eta} e^{-B_\tau} \Omega \| \bigg)    \; . 
\end{align}
Since $d \Gamma (O) \leq \mathcal{N}_+ $ on $\mathcal{F}_{\perp \varphi}^{\leq N} $ and the operators $d\Gamma (O)$ and $\mathcal{N}_+$ commute, we have 
\begin{align}
\|\mathcal{N}_+^k  e^{\lambda d \Gamma (O)} \psi_N \| \leq  \| e^{\lambda_1 \mathcal{N}_+} \psi_N \| \leq C 
\end{align}
for sufficiently small $\lambda_1 >  \vert \lambda  \vert >0$ as a consequence of the exponential bound \eqref{eq:exp-bound} (resp. \cite[Theorem 1.1]{NR}). For the second term we find similarly 
\begin{align}
 \| \mathcal{N}_+^k e^{\lambda d\Gamma (O)} e^{-B_\eta}e^{A_\eta} e^{-B_\tau} \Omega \| \leq  \| e^{\lambda_1 \mathcal{N}_+} e^{-B_\eta}e^{A_\eta} e^{-B_\tau} \Omega \| \leq C 
\end{align}
where the last estimate is a consequence of Lemmas \ref{lemma:bogo-mod} and \ref{lemma:cubic} and $\mathcal{N}_+ \Omega =0$.  \\ 

\textbf{Step 1.2 (Action of $e^{A_\eta}$):} In order to show that the action of the cubic transformation $ e^{A_\eta}$ on the operator $e^{ \lambda d \Gamma (O)}$ is negligible in the large particle limit, we write 
\begin{align}
\langle \Omega,  & e^{-B_\tau} e^{-A_\eta} e^{-B_\eta} \mathcal{N}_+^k e^{\lambda d\Gamma (O)} e^{B_\eta} e^{A_\eta} e^{B_\tau} \Omega \rangle - \langle \Omega,  \, e^{-B_\tau}  e^{-B_\eta} \mathcal{N}_+^k e^{\lambda d\Gamma (O)} e^{B_\eta}  e^{B_\tau} \Omega \rangle \notag \\
=& \int_0^1 ds \frac{d}{ds}\;  \langle \Omega, \; e^{-B_\tau} e^{-s A_\eta} e^{-B_\eta} \mathcal{N}_+^k e^{\lambda d\Gamma (O)} e^{B_\eta} e^{sA_\eta} e^{B_\tau} \Omega \rangle \notag \\
=& - \int_0^1 ds \langle \Omega, \;  e^{-B_\tau} e^{- sA_\eta} \bigg[ A_\eta, \; e^{-B_\eta}\mathcal{N}_+^k e^{\lambda d\Gamma (O)} e^{B_\eta}  \bigg] e^{sA_\eta} e^{B_\tau} \Omega \rangle 
\end{align}
With \eqref{eq:intertwining} we write 
\begin{align}
 \langle \Omega,  & \;  e^{-B_\tau} e^{- sA_\eta} \bigg[ A_\eta, \; e^{-B_\eta} e^{\lambda d\Gamma (O)} e^{B_\eta}  \bigg] e^{sA_\eta} e^{B_\tau} \Omega \rangle  \notag \\ 
=& 2 \Re \langle \Omega,  \;  e^{-B_\tau} e^{-s A_\eta}  A_\eta  e^{-B_\eta} \mathcal{N}_+^k e^{\lambda d\Gamma (O)} e^{B_\eta}  e^{sA_\eta} e^{B_\tau} \Omega \rangle  \; \notag \\
=& 2 N^{-1/2} \sum_{r \in P_H, v \in P_L} \eta_r\sinh (\eta_v)   \notag \\
&\hspace{2.5cm}  \times \Re \langle \Omega,  e^{-B_\tau} e^{- sA_\eta} \mathcal{N}_+^{1/2} b^*_{r+v} b_{-r}^* b_{-v}^* ( \mathcal{N}_+ + 3)^{-1/2} \mathcal{N}_+^k e^{\lambda d\Gamma (O)} e^{B_\eta} e^{sA_\eta} e^{B_\tau} \Omega  \rangle   \notag \\
&+ 2 N^{-1/2} \sum_{r \in P_H, v \in P_L} \eta_r\sinh (\eta_v)  \notag \\
&\hspace{2.5cm}  \times \Re \langle \Omega,  e^{-B_\tau} e^{- sA_\eta} \mathcal{N}_+^{1/2} b_{r+v} b_{-r} b_{-v} ( \mathcal{N}_+ + 1)^{-1} \mathcal{N}_+^k e^{\lambda d\Gamma (O)} e^{B_\eta} e^{sA_\eta} e^{B_\tau} \Omega  \rangle   \notag \\
&+ 2 N^{-1/2} \sum_{r \in P_H, v \in P_L} \eta_r \cosh (\eta_v )  \notag \\
&\hspace{2cm}  \times\Re \langle \Omega,  e^{-B_\tau} e^{- sA_\eta} (\mathcal{N}_+ +4) b^*_{r+v} b^*_{-r}b_v ( \mathcal{N}_+ + 1)^{-1/2} \mathcal{N}_+^k e^{\lambda d\Gamma (O)} e^{B_\eta} e^{sA_\eta} e^{B_\tau} \Omega  \rangle   \notag \\
&+ 2 N^{-1/2} \sum_{r \in P_H, v \in P_L} \eta_r \cosh (\eta_v )  \notag \\
&\hspace{2cm}  \times\Re \langle \Omega,  e^{-B_\tau} e^{- sA_\eta} (\mathcal{N}_+ + 2) b^*_v b_{r+v} b_{-r} ( \mathcal{N}_+ + 1)^{-1} \mathcal{N}_+^k e^{\lambda d\Gamma (O)} e^{B_\eta} e^{sA_\eta} e^{B_\tau} \Omega  \rangle   \; . 
\end{align}
With similar computations as in \eqref{eq:cubic-bound1}-\eqref{eq:cubic-bound2}, we find 
\begin{align}
\vert \langle \Omega,  & \;  e^{-B_\tau} e^{- A_\eta} \bigg[ A_\eta, \; e^{-B_\eta} e^{\lambda d\Gamma (O)} e^{B_\eta}  \bigg] e^{sA_\eta} e^{B_\tau} \Omega \rangle \vert \notag \\ 
\leq& C N^{-1/2} \|( \mathcal{N}_+ + 1)^{3/2} e^{B_\tau} e^{A_\eta} \Omega \| \; \|\mathcal{N}_+^k e^{\lambda d \Gamma (O)} e^{B_\eta}  e^{sA_\eta} e^{B_\tau} \Omega \| \; . 
\end{align}
For the first term, the estimates \eqref{eq:bogo-mod-moment} and \eqref{eq:cubic-moment} imply $\|( \mathcal{N}_+ + 1)^{3/2} e^{B_\tau} e^{A_\eta} \Omega \| \leq C$. For the second term, we estimate for $\lambda_1 > \vert \lambda \vert > 0$ that $\| \mathcal{N}_+^k e^{\lambda d \Gamma (O) }e^{B_\eta}  e^{sA_\eta} e^{B_\tau} \Omega  \| \leq  C_k \| e^{C \lambda_1 \mathcal{N}_+ }e^{B_\eta}  e^{sA_\eta} e^{B_\tau} \Omega  \| $ and then find with  Lemmas \ref{lemma:bogo-mod}, \ref{lemma:cubic} 
\begin{align}
\|\mathcal{N}_+^k e^{\lambda d \Gamma (O) }e^{B_\eta}  e^{sA_\eta} e^{B_\tau} \Omega  \| \leq C_k \;  
\end{align}
for sufficiently small $\lambda_1 >0$ that finally yield
\begin{align}
\vert \langle \Omega,  & e^{-B_\tau} e^{-A_\eta} e^{-B_\eta} \mathcal{N}_+^k e^{\lambda d\Gamma (O)} e^{B_\eta} e^{A_\eta} e^{B_\tau} \Omega \rangle - \langle \Omega,  \, e^{-B_\tau}  e^{-B_\eta} \mathcal{N}_+^k e^{\lambda d\Gamma (O)} e^{B_\eta}  e^{B_\tau} \Omega \rangle  \vert \leq CN^{-1/2}  
\end{align}
\end{proof}

\textbf{Step 1.3 (Asymptotic standard Bogoliubov transformation):} In the last step, we show that the action of the two modified Bogoliubov transformations w.r.t. the kernels $\eta_p$ and $\tau_p$ defined in \eqref{def:eta} resp. \eqref{def:tau} on the operator $\mathcal{N}_+^k e^{\lambda d \Gamma (O)}$ are asymptotically equivalent to the action of one effective standard Bogoliubov transformation w.r.t. the kernel $\nu_p$ defined in \eqref{def:sigma}. This observation is based on the observation that 
\begin{align}
\vert \eta_p + \tau_p - \nu_p \vert \leq  C N^{-1} \label{eq:approx-kernel}
\end{align}
whose proof is given in \cite[Section 3]{RS}. 

\begin{lemma}
\label{lemma:step3} Under the same assumptions as in Theorem \ref{thm:mgf}, there exists $C>0$ such that 
\begin{align}
\vert \langle \Omega, e^{-B_\tau} e^{-B_\eta} \mathcal{N}_+^k e^{\lambda d \Gamma (O)} e^{B_\eta} e^{B_\tau} \Omega \rangle - \langle \Omega, e^{-K_\nu} \mathcal{N}_+^k e^{\lambda d \Gamma (O)} e^{K_\nu} \Omega \rangle \vert \leq CN^{-1} \; . 
\end{align}
\end{lemma}

\begin{proof}
The goal is to compare the operator 
\begin{align}
\label{def:calA}
\mathcal{A} :=& e^{-K_\nu} d\Gamma ( O) e^{K_\nu} \notag \\
&=  \sum_{p,q \in \pi_+^*}  O_{p,q}   \big[ \cosh( \nu_p) a_p^* + \sinh( \nu_p) a_{-p} \big]  \big[ \cosh( \nu_q) a_q^* + \sinh( \nu_q) a_{-q} \big] 
\end{align}
with $ \mathcal{B}  = e^{-B_\tau}e^{-B_\eta} d\Gamma ( O) e^{B_\eta}e^{B_\tau}$ that we compute using \eqref{eq:bogo-mod-action} and properties of the hyperbolic functions 
\begin{align}
\label{def:calB}
\mathcal{B} :=&  e^{-B_\tau}e^{-B_\eta} d\Gamma ( O) e^{B_\eta}e^{B_\tau} \notag \\
=& \sum_{p,q \in \pi_+^*} O_{p,q} \big[ \cosh( \eta_p + \tau_p) b_p^* + \sinh( \eta_p + \tau_p) b_{-p} \big]  \big[ \cosh( \eta_q + \tau_q) b_q^* + \sinh( \eta_q + \tau_q) b_{-q} \big]  \notag \\
&+ e^{-B_\tau}\mathcal{E}_{B_\eta} e^{B_\tau} + \mathcal{E}_{B_\tau} \; 
\end{align}
where the errors $\mathcal{E}_{B_\alpha}$ are for $\alpha \in \ell^2( \pi_+^*)$ given by 
\begin{align}
\mathcal{E}_{B_\alpha} =& \sum_{p,q \in \pi_+^*} \big(  O_{p,q}  \big[  \cosh(\alpha_p) b_p^* + \sinh( \alpha_p) b_{-p} \big] d_{\alpha_q}  + {\rm h.c} \big)   + \sum_{p,q \in \pi_+^*} O_{p,q} d^*_{\alpha_p}d_{\alpha_q}  
\end{align}
and satisfy by \eqref{eq:bounds-d} since $\sup_{p,q} \vert O_{p,q} \vert  \leq C $ 
\begin{align}
\label{def:E_B}
\vert \langle \xi_1, \mathcal{E}_{B_\alpha}\xi_2 \rangle \vert  \leq CN^{-1/2} \| ( \mathcal{N}_+ + 1)^{3/2} \xi_1 \| \; \| \xi_2 \| \; 
\end{align}
for any $\xi_1,\xi_2 \in \mathcal{F}_{\perp \varphi}^{\leq N}$.  Furthermore, we introduce the notation 
\begin{align}
\label{def:A1}
\mathcal{A}_1 := e^{-K_\nu} \mathcal{N}_+  e^{K_\nu} , \quad \text{and} \quad \mathcal{B}_1 := e^{-B_\tau} e^{-B_\eta} \mathcal{N}_+ e^{B_\eta} e^{B_\tau} \; . 
\end{align}

We remark that while $\mathcal{A}, \mathcal{A}_1$ are operators on the full bosonic Fock space (the domain of the standard Bogoliubov transform), the operator $\mathcal{B}$ acts on the truncated Fock space only (the domain of the modified Bogoliubov transform). Thus, in order to compare $\mathcal{A}, \mathcal{A}_1$ and $\mathcal{B}, \mathcal{B}_1$ we split an element $\psi \in \mathcal{F}$ of the full bosonic Fock space into $\psi = \mathds{1}_{\mathcal{N}_+ \leq N}\psi + \mathds{1}_{\mathcal{N}_+ > N}\psi $ and then write 
\begin{align}
\langle \Omega,  & \mathcal{B}_1 e^{\mathcal{B}} \Omega \rangle - \langle \Omega, \mathcal{A}_1 e^{\mathcal{A}} \Omega \rangle  \notag \\
 =& \langle \Omega,   \mathcal{B}_1 e^{\mathcal{B}} \Omega \rangle - \langle \Omega, e^{\mathcal{A}/2}\mathds{1}_{\mathcal{N} \leq N} \mathcal{A}_1 \mathds{1}_{\mathcal{N} \leq N} e^{\mathcal{A}/2} \Omega \rangle \notag \\
 &-  \langle \Omega, e^{\mathcal{A}/2} \mathcal{A}_1 \mathds{1}_{\mathcal{N} > N} e^{\mathcal{A}/2} \Omega \rangle - \langle \Omega, e^{\mathcal{A}/2} \mathds{1}_{\mathcal{N} > N} \mathcal{A}_1 \mathds{1}_{\mathcal{N} < N} e^{\mathcal{A}/2} \Omega \rangle \; .  \label{eq:difference}
\end{align} 
We bound the third term of the r.h.s. of \eqref{eq:difference}, using $\mathds{1}_{\mathcal{N} > N} \leq \mathcal{N}_+ /N  $ on $\mathcal{F}$ and find 
\begin{align}
 \vert\langle \Omega, e^{\mathcal{A}/2}\mathcal{A}_1\mathds{1}_{\mathcal{N} > N} e^{\mathcal{A}/2} \Omega \rangle \vert \leq&  CN^{-1} \| \mathcal{A}_1 e^{\mathcal{A}/2} \Omega \| \; \bigg\| \frac{\mathcal{N}}{N} \mathds{1}_{\mathcal{N} > N} e^{\mathcal{A}/2} \Omega \bigg\| \; . 
\end{align}
From the definition of $\mathcal{A}_1$ in \eqref{def:A1} and \eqref{eq:bogo-mom} we find 
\begin{align}
\vert \langle \Omega, e^{\mathcal{A}/2}\mathcal{A}_1\mathds{1}_{\mathcal{N} > N} e^{\mathcal{A}/2} \Omega \rangle  \vert  \leq&  CN^{-1} \| ( \mathcal{N} + 1) e^{\mathcal{A}/2} \Omega \| 
\end{align}
Recalling the definition of $\mathcal{A}$ in \eqref{def:calA}, we find with \eqref{eq:bogo-mom} and Lemma \ref{lemma:bogo-standard}
\begin{align}
\vert \langle  \Omega, e^{\mathcal{A}/2} \mathcal{A}_1 \mathds{1}_{\mathcal{N} > N} e^{\mathcal{A}/2} \Omega \rangle \vert
\leq&   C N^{-1} \| (\mathcal{N} + 1) e^{-K_\nu} e^{\lambda d\Gamma (O)} e^{K_\nu}\Omega \|^2  \notag \\
\leq&   C N^{-1} \| (\mathcal{N}_+ + 1) e^{\lambda d\Gamma (O)} e^{K_\nu}\Omega \|^2   \; . 
\end{align}
For some $\lambda_1 > \vert \lambda \vert $ we continue with 
\begin{align}
\vert \langle  \Omega, e^{\mathcal{A}/2} \mathcal{A}_1 \mathds{1}_{\mathcal{N} > N} e^{\mathcal{A}/2} \Omega \rangle \vert
\leq&  C  N^{-1} \| e^{\lambda_1 \mathcal{N}_+} e^{K_\nu}\Omega \|^2 \leq C  N^{-1}  \; .  \label{eq:Nbig}
\end{align}
where we concluded by Lemma \ref{lemma:bogo-standard}. The forth term of the r.h.s. of \eqref{eq:difference} can be estimated similarly. 

Next we show that the difference of the first two terms of the r.h.s. of \eqref{eq:difference} vanishes in the large particle limit, too. To this end, we write 
\begin{align}
\langle \Omega,  &  e^{\mathcal{B}/2}\mathcal{B}_1e^{\mathcal{B}/2}  \Omega \rangle - \langle \Omega, e^{\mathcal{A}/2} \mathds{1}_{\mathcal{N} \leq N} \mathcal{A}_1 \mathds{1}_{\mathcal{N} \leq N} e^{\mathcal{A}/2} \Omega \rangle  \notag \\
=& \langle \Omega,   e^{\mathcal{B}/2}\mathcal{B}_1e^{\mathcal{B}/2}  \Omega \rangle - \langle \Omega, e^{\mathcal{A}/2} \mathds{1}_{\mathcal{N} \leq N} \mathcal{B}_1 \mathds{1}_{\mathcal{N} \leq N} e^{\mathcal{A}/2} \Omega \rangle \notag \\
&\quad + \langle \Omega, e^{\mathcal{A}/2} \mathds{1}_{\mathcal{N} \leq N} \big( \mathcal{A}_1 - \mathcal{B}_1 \big) \mathds{1}_{\mathcal{N} \leq N} e^{\mathcal{A}/2} \Omega \rangle \notag \\
=& \int_0^s ds \frac{d}{ds} \langle \Omega, e^{s\mathcal{A}/2}\mathds{1}_{\mathcal{N} \leq N} e^{ (1-s) \mathcal{B}/2} \mathcal{B}_1 e^{ (1-s) \mathcal{B}/2} \mathds{1}_{\mathcal{N} \leq N}  e^{s \mathcal{A}/2}\Omega \rangle \notag \\
&\quad + \langle \Omega, e^{\mathcal{A}/2} \mathds{1}_{\mathcal{N} \leq N} \big( \mathcal{A}_1 - \mathcal{B}_1 \big) \mathds{1}_{\mathcal{N} \leq N} e^{\mathcal{A}/2} \Omega \rangle \notag \\
=&  \Re \int_0^1 ds  \langle \Omega, e^{s\mathcal{A}/2} \bigg( \big[ \mathcal{A}, \mathds{1}_{\mathcal{N} \leq N} \big] + \mathds{1}_{\mathcal{N}_+ \leq N} (\mathcal{A} - \mathcal{B} ) \bigg) e^{ (1-s) \mathcal{B}/2} \mathcal{B}_1 e^{ (1-s) \mathcal{B}/2} \mathds{1}_{\mathcal{N} \leq N}  e^{s \mathcal{A}/2}\Omega \rangle  \notag \\
&\quad + \langle \Omega, e^{\mathcal{A}/2} \mathds{1}_{\mathcal{N} \leq N} \big( \mathcal{A}_1 - \mathcal{B}_1 \big) \mathds{1}_{\mathcal{N} \leq N} e^{\mathcal{A}/2} \Omega \rangle \; . 
\end{align}
On the one hand, we have by definition of $\mathcal{A}, \mathcal{A}_1$ and $\mathcal{B}, \mathcal{B}_1$ in \eqref{def:calA}, \eqref{def:A1} resp. \eqref{def:calB}, recalling that $b_p = \sqrt{1- \mathcal{N}_+ /N}$ from \eqref{def:b}, the estimates \eqref{def:E_B} and \eqref{eq:approx-kernel} 
\begin{align}
\| ( \mathcal{A}_1 - \mathcal{B}_1 ) \psi \|, \; \| ( \mathcal{A} - \mathcal{B} ) \psi \| \leq CN^{-1} \| ( \mathcal{N}_+ + 1)^{5/2} \psi \|  \; . \label{eq:diff1}
\end{align}
On the other hand, since 
\begin{align}
\big[ &  \mathcal{A}, \mathds{1}_{\mathcal{N} \leq N} \big]\notag \\
 =& \sum_{p,q \in \pi_+^*} O_{p,q} \cosh(\nu_p) \sinh(\nu_q) a_p^*a^*_{-q}  ( \mathds{1}_{\mathcal{N} \leq N}  - \mathds{1}_{\mathcal{N}+2 \leq N} ) \notag \\
 &+ \sum_{p,q \in \pi_+^*} O_{p,q}  \sinh(\nu_p) \cosh(\nu_q) a_{-p}a_{q}  ( \mathds{1}_{\mathcal{N} \leq N}  - \mathds{1}_{\mathcal{N} -2 \leq N} )\notag \\
 =& \sum_{p,q \in \pi_+^*} O_{p,q} \big( \cosh(\nu_p) \sinh(\nu_q) a_p^*a^*_{-q}   \mathds{1}_{N-1 \leq \mathcal{N}_+ \leq N} - \sinh(\nu_p) \cosh(\nu_q) a_{-p}a_{q}  \mathds{1}_{N+1 \leq \mathcal{N}_+ \leq N +2} \big)
\end{align}
With $\mathds{1}_{N-1 \leq \mathcal{N}_+ \leq N} \leq C N^{-1} \mathcal{N}_+ $ and $\mathds{1}_{N+1 \leq \mathcal{N}_+ \leq N +2} \leq CN^{-1} \mathcal{N}$, we find that 
\begin{align}
\| \big[   \mathcal{A}, \mathds{1}_{\mathcal{N} \leq N} \big] \psi \| \leq CN^{-1} \| ( \mathcal{N}_+ + 1) \psi \| \; . \label{eq:diff2}
\end{align}
Summarizing \eqref{eq:diff1} and \eqref{eq:diff2}, we arrive again with \eqref{eq:bogo-mod-moment} and Lemma \ref{lemma:bogo-mod} at 
\begin{align}
\vert \langle \Omega,  &  \mathcal{B}_1 e^{\mathcal{B}} \Omega \rangle - \langle \Omega, e^{\mathcal{A}/2}\mathds{1}_{\mathcal{N} \leq N} \mathcal{A}_1 \mathds{1}_{\mathcal{N} \leq N}  e^{\mathcal{A}/2} \Omega \rangle \vert \leq C N^{-1}  \; 
\end{align}
leading with \eqref{eq:Nbig} to
\begin{align}
\vert \langle \Omega, \mathcal{A}_1 e^{\mathcal{A}} \Omega \rangle - \langle \Omega, \mathcal{B}_1
e^{\mathcal{B}} \Omega \rangle \vert \leq C N^{-1} 
\end{align}
that concludes the last step for the proof of Lemma \ref{lemma:approx}. 
\end{proof}

\subsection{Step 2: Asymptotic generating function}
\label{sec:proof-lemma-mgf}

In this section, we prove Lemma \ref{lemma:mfg}, i.e. we explicitly compute the asymptotic generating function that is a generating of a quasi-free state. 

\begin{proof}[Proof of Lemma \ref{lemma:mfg}]

 Since $\mathcal{N} = d\Gamma (\mathds{1}) = \sum_{p \in \pi_+^*} a_p^*a_p$, the calculations for the first formula \eqref{eq:mgf-N} will turn out the be a special case of the second one \eqref{eq:mgf-O}. For this reason we formulate the beginning of the proof for general operators $O$ on $\ell^2( \pi_+^*)$ and later restrict to the special (easier) case $ O_{p,q} = \delta_{p,q}$ (referring to \eqref{eq:mgf-N}) first, to use the special cases's ideas to prove the general case (i.e. \eqref{eq:mgf-O}). 

The proof's goal is to show that the function 
\begin{align}
G: [-\lambda_0,\lambda_0] \rightarrow \mathbb{R}, \quad G( \lambda) = \langle \Omega, e^{K_\nu} e^{\lambda d \Gamma (O)} e^{K_\nu} \Omega \rangle 
\end{align}
for sufficiently small $\lambda_0 \in \mathbb{R}$ is the unique solution of a differential equation. To this end, we observe that $G(0) = 1$ and compute 
\begin{align}
G'( \lambda ) =  \langle \Omega, e^{-K_\nu} d\Gamma (O) e^{\lambda d\Gamma (O)} e^{K_\nu} \Omega \rangle = \sum_{p,q \in \pi_+^*} O_{p,q} \; \langle \Omega, e^{-K_\nu} a_p^*a_q e^{\lambda d \Gamma (O)} e^{K_\nu} \Omega \rangle \; .  \label{eq:G1}
\end{align}
In the following, we aim for an explicit expression of the r.h.s. of \eqref{eq:G1} in terms of $G( \lambda)$. We base the calculations on the explicit formulas \eqref{eq:bogo-action} and \eqref{lemma:BCH} for the conjugation of creation and annihilation operators with $e^{K_\nu}$ resp. $e^{\lambda d \Gamma (O)}$ and the fact that the vacuum is an eigenstate of the annihilation operator with eigenvalue zero. In fact, with the short-hand notation
\begin{align}
s_p := \sinh( \nu_p) , \quad \text{and} \quad c_p := \cosh (\nu_p ), 
\end{align}
we find from \eqref{eq:bogo-action} 
\begin{align}
e^{-K_\nu} a_p^*a_q e^{K_\nu} =&  s_p s_q a_{-q}^*a_{-p} + c_p c_q a_p^*a_q  + s_q c_p  a^*_p a_{-q}^* + s_p c_q a_{-p} a_{q} +  s_p s_q \delta_{p,q} \; . 
\end{align}
Thus, on the one hand, by commuting $a_p^*a_q$ in the scalar product on the r.h.s. of \eqref{eq:derivG1} to the left, we get with $a_p \Omega =0$ for all $p \in \Lambda_+^*$ 
\begin{align}
\sum_{p,q \in \pi_+^*}O_{p,q} & \langle \Omega,  e^{-K_\nu} a_p^*a_q  e^{\lambda d\Gamma (O)} e^{K_\nu} \Omega \rangle \notag \\ 
=& \sum_{p,q \in \pi_+^*} s_p c_q O_{p,q} \langle \Omega, a_{-p} a_{q} e^{-K_\nu}   e^{\lambda d\Gamma (O)} e^{K_\nu} \Omega \rangle  +\sum_{p\in \pi_+^*} O_{p,p} s_p^2 \notag \\
=&  \sum_{p,q \in \pi_+^*} O_{p,q} s_p c_q F_{p,q}( \lambda)  + \sum_{p \in \pi_+^*} O_{p,p} s_p^2 \label{eq:derivG1}
\end{align}
where we introduced the notation 
\begin{align}
F_{p,q} (\lambda)  = \langle \Omega, a_{-p} a_{q} e^{-K_\nu}   e^{\lambda d\Gamma (O)} e^{K_\nu} \Omega \rangle  \; . 
\end{align}
On the other hand, commuting the operator $d\Gamma (O)$ in the scalar product on the r.h.s. of \eqref{eq:G1} to the right, we find 
\begin{align}
\sum_{p,q \in \pi_+^*}O_{p,q} &\langle \Omega,   e^{-K_\nu} a_p^*a_q e^{\lambda d\Gamma (O)} e^{K_\nu} \Omega \rangle  \notag \\
=&  \sum_{p,q \in \pi_+^*}O_{p,q} c_p s_q \langle \Omega, e^{-K_\nu}   e^{\lambda d\Gamma (O)} e^{K_\nu} a_{p}^* a^*_{-q} \Omega \rangle + \sum_{p \in \pi_+^*}O_{p,p} s_p^2 \notag \\
=& \sum_{p,q \in \pi_+^*}O_{p,q}  c_p s_q  \overline{F_{q,p}(\lambda)}  +\sum_{p\in \pi_+^*}O_{p,p}  O_{p,p} s_p^2  \; . \label{eq:derivG2}
\end{align}
Comparing \eqref{eq:derivG1} with \eqref{eq:derivG2}, we observe that 
\begin{align}
\label{eq:F-overlineF}
c_q s_p F_{p,q} (\lambda) = c_p s_q \overline{F_{q,p} (\lambda)} \quad \Longleftrightarrow\quad  \overline{F_{q,p}(\lambda)} = \frac{\tau_q}{\tau_p} F_{p,q}(\lambda) , \quad \text{with} \quad \tau_p = \tanh( \nu_p )
\end{align}
for all $p,q \in {\rm supp} (O)=: \lbrace p,q \in \pi_+^* : O_{p,q} \not= 0 \rbrace$. Otherwise, i.e. for all $p,q \in \pi_+^*$ such that $O_{p,q} =0$, we have 
\begin{align}
s_p c_q F_{p,q} ( \lambda) = \langle \Omega,  e^{-K_\nu} a_p^*a_q  e^{\lambda d\Gamma (O)} e^{K_\nu} \Omega \rangle = \langle \Omega,  e^{-K_\nu}   e^{\lambda d\Gamma (O)}a_p^*a_q  e^{K_\nu} \Omega \rangle = c_p s_q \overline{F_{q,p} (\lambda)} ,
\end{align}
and thus, \eqref{eq:F-overlineF} holds for all $p,q \in \pi_+^*$. This identity will be useful later when we aim for an explicit expression of the operator $F( \lambda)$ in terms of $G(\lambda)$. The idea is to show that the kernel $F_{-p,q} (\lambda)$ is a the unique fixed point of a linear operator that we can construct explicitly. The calculations are again based on the properties of the Bogoliubov transform \eqref{eq:bogo-action}. We start with the observation that with $c_{-p} = c_p, s_{-p} = s_p$ for all $p \in \pi_+^*$ we can write 
\begin{align}
F_{p,q} (\lambda) =& \big\langle \Omega,  e^{-K_\nu} \big( c_p c_q a_{-p} a_{q} -  s_p c_q a_{p}^*a_{q} - c_p s_q a_{-p} a_{-q}^* + s_p s_q a_{p}^*a_{-q}^* \big)  e^{\lambda d\Gamma (O)} e^{K_\nu} \Omega \big\rangle  \notag \\
=& \big\langle \Omega,  e^{-K_\nu} \big( c_p c_q a_{-p} a_{q} +  s_q s_p  a_{p}^*a_{-q}^*  -  s_p c_q a_{p}^*a_{q}  - c_ps_q a_{-q}^*a_{-p} \big)  e^{\lambda d\Gamma (O)} e^{K_\nu} \Omega \big\rangle  \notag \\
&- c_ps_q \delta_{p,q} G( \lambda)  \notag \\
=& \mathrm{I} +  \mathrm{II} + \mathrm{III} + \mathrm{IV} - c_ps_q \delta_{p,q} G( \lambda)  \; . \label{eq:Fp1}
\end{align}
Next, we compute all the four terms $\mathrm{I}-\mathrm{IV}$ of the r.h.s. of \eqref{eq:Fp1} separately. We start with the first term for which we commute the pair of creation and annihilation operators to the right and use that $a_k \Omega =0$ for all $k \in \pi_+^*$. We find with Lemma \ref{lemma:BCH} 
\begin{align}
\mathrm{I} =& c_pc_q  \big\langle \Omega,   e^{-K_\nu}  a_{-p} a_{q}   e^{\lambda d\Gamma (O)} e^{K_\nu} \Omega \big\rangle \notag \\
 =& c_pc_q  \big\langle \Omega,  e^{-K_\nu}  e^{\lambda d\Gamma (O)} a (e^{\lambda O_{-p}} ) a(e^{\lambda O_{q}}) e^{K_\nu} \Omega \big\rangle 
\end{align}
and furthermore with \eqref{eq:bogo-action-h} using that $s_p,c_p \in \mathbb{R}$, $s_p = s_{-p}$ for all $p \in \pi_+^*$ and the notations \eqref{eq:not-O-p}, \eqref{eq:not-minus} 
\begin{align}
\mathrm{I} 
 =& c_pc_q \big\langle \Omega,  e^{-K_\nu}  e^{\lambda d\Gamma (O)}e^{K_\nu}   \bigg[   a \big( c e^{\lambda O_{-p}} \big) a\big(c e^{\lambda O_{q}}\big)  +  a \big(ce^{\lambda O_{-p}} \big)a^*\big( se^{\lambda \overline{(O_{q})^{-}}}\big)  \bigg]  \Omega \big\rangle \notag \\
&+ c_pc_q\big\langle \Omega,  e^{-K_\nu}  e^{\lambda d\Gamma (O)}e^{K_\nu} \bigg[   a^* \big( s e^{\lambda \overline{(O_{-p})^-}} \big) a\big( c e^{\lambda O_{q}}\big) +   a^* \big(s e^{\lambda \overline{O_{-p})^-}} \big) a^*\big( se^{\lambda \overline{(O_{q})^-}}\big) \bigg]    \Omega \big\rangle \notag \\
=&  c_pc_q  \sum_{k \in \pi_+^*} c_k s_k e^{\lambda \overline{O}_{k,-p}}  e^{\lambda \overline{O}_{-k,q}}  G( \lambda) \notag \\
&+c_pc_q  \sum_{k, \ell \in\pi_+^*}  s_k s_\ell  e^{\lambda \overline{O}_{-k,-p}}  e^{\lambda \overline{O}_{\ell,q}}     \big\langle \Omega,  e^{-K_\nu}  e^{\lambda d\Gamma (O)}e^{K_\nu}    a^*_k a^*_{-\ell} \Omega \big\rangle \notag \notag \\
=& c_pc_q  \langle c e^{\lambda O_{-p}}, s  e^{\lambda \overline{O}^-_{q}} \rangle_{\ell^2}  G( \lambda) + c_pc_q  \sum_{k, \ell \in \pi_+^*}  s_k s_\ell  e^{\lambda \overline{O}_{-k,-p}}  e^{\lambda \overline{O}_{\ell,q}}    \overline{F_{\ell,k}( \lambda)}
\end{align}
Now with the property \eqref{eq:F-overlineF}, we can express the r.h.s. in terms of $F_{k,\ell}( \lambda)$ and find 
\begin{align}
\mathrm{I} 
=& c_pc_q  \langle c e^{\lambda O_{-p}}, s  e^{\lambda \overline{O}^-_{q}} \rangle_{\ell^2}  G( \lambda) + c_pc_q  \sum_{k, \ell \in \pi_+^*}  \frac{s_k}{\tau_k} s_\ell\tau_\ell  e^{\lambda \overline{O}_{-k,-p}}  e^{\lambda \overline{O}_{\ell,q}}      F_{k,\ell}( \lambda) \label{eq:term1-final}
\end{align}
Next, we compute the second term of the r.h.s. of \eqref{eq:Fp1} with similar ideas as the first one, and find with Lemma \ref{lemma:BCH} 
\begin{align}
\mathrm{II} =& s_p s_q   \big\langle \Omega,  e^{-K_\nu} a^*_{-p} a^*_q e^{\lambda d \Gamma (O)} e^{K_\nu} \Omega \big\rangle  \notag \\
=& s_p s_q  \big\langle \Omega,   e^{-K_\nu} e^{\lambda d \Gamma (O)} a^*\big(e^{-\lambda O_{-p}}\big) a^*\big( e^{-\lambda O_q} \big)  e^{K_\nu} \Omega \big\rangle \notag \\
=& s_p s_q  \big\langle \Omega,   e^{-K_\nu} e^{\lambda d \Gamma (O)} e^{K_\nu}\bigg[  a^*\big(ce^{-\lambda O_{-p}}\big) a^*\big( ce^{-\lambda O_q} \big)  + a^*\big(ce^{-\lambda O_{-p}}\big) a\big( se^{-\lambda \overline{(O_q)^-}} \big)  \bigg] \Omega \big\rangle \notag \\
&+ s_p s_q  \big\langle \Omega,   e^{-K_\nu} e^{\lambda d \Gamma (O)} e^{K_\nu}\bigg[  a\big(se^{-\lambda \overline{O_{-p}^-}}\big) a^*\big( ce^{-\lambda O_q} \big)  + a\big(se^{-\lambda \overline{O_{-p}^-}}\big) a\big( se^{-\lambda \overline{(O_q)^-}} \big)  \bigg] \Omega \big\rangle \notag \\
=& s_p s_q \langle s e^{-\lambda \overline{O}_{-p}^-}, c e^{-\lambda O_q} \rangle G( \lambda) + s_p s_q \sum_{k,\ell \in \pi_+^*} c_kc_\ell e^{-\lambda O_{k,-p}} e^{-\lambda O_{-\ell,q}} \overline{F_{\ell,k} (\lambda)} \notag \\
=& s_p s_q \langle s e^{-\lambda \overline{O}_{-p}^-}, c e^{-\lambda O_q} \rangle G( \lambda) + s_p s_q \sum_{k,\ell \in \pi_+^*} \frac{c_k}{\tau_k} c_\ell \tau_\ell e^{-\lambda O_{k,-p}} e^{-\lambda O_{-\ell,q}}  F_{k,\ell}( \lambda)
\label{eq:term2-final}
\end{align}  
Similarly, we find for the third term of the r.h.s. of \eqref{eq:Fp1} 
\begin{align}
\mathrm{III} =& - s_pc_q   \langle\Omega, e^{-K_\nu} a_{-p}^*a_{-q} e^{\lambda d \Gamma (O)} e^{K_\nu} \Omega  \rangle \notag \\
=& - s_p c_q \langle \Omega, e^{-K_\nu}  e^{\lambda d \Gamma (O)} a^*\big( e^{-\lambda O_{-p}}\big) a\big( e^{\lambda O_{-q}}\big) e^{K_\nu} \Omega  \rangle \notag \\
=& - s_p c_q \langle \Omega, e^{-K_\nu}  e^{\lambda d \Gamma (O)} e^{K_\nu}  \bigg[ a^*\big( c e^{-\lambda O_{-p}}\big) a\big( ce^{\lambda O_{-q}}\big) + a^*\big( ce^{-\lambda O_{-p}}\big) a^*\big( se^{\lambda \overline{(O_{-q})^-}}\big) \bigg] \Omega  \rangle \notag \\
&- s_p c_q \langle \Omega, e^{-K_\nu}  e^{\lambda d \Gamma (O)} e^{K_\nu}  \bigg[ a\big( s e^{-\lambda \overline{(O_{-p})^-}}\big) a\big( ce^{\lambda O_{-q}}\big) + a\big( se^{-\lambda \overline{(O_{-p})^-}}\big) a^*\big( se^{\lambda \overline{(O_{-q})^-}}\big) \bigg] \Omega  \rangle \notag \\
&= - s_pc_q \overline{\langle s e^{-O_{-p}}, s e^{\lambda O_{-q}} \rangle} G( \lambda) -s_pc_q \sum_{k,\ell \in \pi_+^*} c_ks_\ell e^{-\lambda O_{k,-p}} e^{\lambda \overline{O}_{\ell, -q}} \overline{F_{\ell,k} ( \lambda)} \notag \\
&= - s_pc_q \overline{\langle s e^{-O_{-p}}, s e^{\lambda O_{-q}} \rangle} G( \lambda) -s_pc_q \sum_{k,\ell \in \pi_+^*} \frac{c_k}{\tau_k} s_\ell \tau_\ell e^{-\lambda O_{k,-p}} e^{\lambda \overline{O}_{\ell, -q}}    F_{k,\ell}( \lambda)  \label{eq:term3-final}
\end{align}
and for the forth term of the r.h.s. of \eqref{eq:Fp1} (be replacing $p,q$ with $-q,-p$ in the previous formula) 
\begin{align}
\mathrm{IV} 
&= - s_qc_p \overline{\langle s e^{-O_{q}}, s e^{\lambda O_{p}} \rangle} G( \lambda) -s_qc_p \sum_{k,\ell \in \pi_+^*} \frac{c_k}{\tau_k} s_\ell \tau_\ell e^{-\lambda O_{k,q}} e^{\lambda \overline{O}_{\ell, p}}    F_{k,\ell}( \lambda)  \label{eq:term4-final}
\end{align}
Summarizing \eqref{eq:term1-final},\eqref{eq:term2-final}, \eqref{eq:term3-final} and \eqref{eq:term4-final}, we thus find from \eqref{eq:Fp1}
\begin{align}
F_{p,q}(\lambda) =&  A_{p,q} ( \lambda)  G(\lambda) +   \sum_{k, \ell \in \pi_+^*} D_{p,q, k,\ell} (\lambda) \;  F_{k, \ell} ( \lambda)  \label{eq:Fp2}
\end{align}
with 
\begin{align}
A_{p,q} ( \lambda)  =& - c_p s_q \delta_{p,q} + c_pc_q  \langle c e^{\lambda O_{-p}}, s  e^{\lambda \overline{O}^-_{q}} \rangle_{\ell^2} +  s_p s_q \langle s e^{-\lambda \overline{O}_{-p}^-}, c e^{-\lambda O_q} \rangle \notag \\
&- s_pc_q \overline{\langle s e^{-O_{-p}}, s e^{\lambda O_{-q}} \rangle} - c_ps_q \overline{\langle s e^{-O_{q}}, s e^{\lambda O_{p}} \rangle} \label{def:Apq-1}
\end{align}
and
\begin{align}
\label{def:D-1}
&D_{p,q, k,\ell} ( \lambda) \notag \\
&\quad =   c_p c_k  c_q  s_\ell \tau_\ell \big(  e^{\lambda \overline{O}_{-k,-p}}    e^{\lambda \overline{O}_{\ell,q}} + e^{-\lambda O_{k,-p}} e^{-\lambda O_{-\ell,q}} - e^{-\lambda O_{k,-p}}  e^{\lambda \overline{O}_{\ell, -q}}  - e^{-\lambda O_{k,p}} e^{\lambda \overline{O}_{\ell, q}} \big) \; . 
\end{align}

Next, we show that \eqref{eq:Fp2} has a unique solution both, for the special choice $O_{p,q} = \delta_{p,q}$ (referring to the special case $d\Gamma (O) = \mathcal{N}_+$ in \eqref{eq:mgf-N}) and the case for general $O$ such that $O_{p,q} \in \ell^2( \pi_+^*) \times \ell^2( \pi_+^*)$ (referring to \eqref{eq:mgf-O}). We treat both cases separately. 

\subsubsection*{\textbf{Proof of \eqref{eq:mgf-N}:}} We start with the easier case $\mathcal{N} = d \Gamma ( \mathds{1})$, i.e. $O_{p,q} = \delta_{p,q}$. Then \eqref{eq:Fp2} reduces to 
\begin{align}
F_p ( \lambda) =  c_p s_p \big( 2 c_p^2 (\cosh( 2\lambda) - 1)  - e^{-2\lambda} + 1\big) G( \lambda) + 2 s_p^2c_p^2 (\cosh(2\lambda) - 1)  F_p (\lambda) 
\end{align}
that we can write as 
\begin{align}
 \big( 1 -2 s_p^2c_p^2 (\cosh(2\lambda) - 1)  \big) F_p (\lambda) = c_p s_p \big( 2 c_p^2 (\cosh( 2\lambda) - 1)  - e^{-2\lambda} + 1\big) G( \lambda) \; . 
\end{align}
For $\lambda$ sufficiently small such that 
\begin{align}
\cosh( 2 \lambda) - 1 < \frac{1}{2 \big( \sup_{p \in \pi_+^*} s_p^2 \big) \big(\sup_{q \in \pi_+^*} c_q^2 \big)}
\end{align}
(note that the r.h.s. is bounded by \eqref{def:nu}), we have 
\begin{align}
F_p (\lambda) = \frac{c_ps_p}{  1 -2 s_p^2c_p^2 (\cosh(2\lambda) - 1)  } \big( 2 c_p^2 (\cosh( 2\lambda) - 1)  - e^{-2\lambda} + 1\big) G( \lambda) \; . 
\end{align}
Plugging this back into \eqref{eq:derivG1}, we find 
\begin{align}
G'(\lambda) = \sum_{p \in \pi_+^*} \frac{c_p^2s_p^2}{  1 -2 s_p^2c_p^2 (\cosh(2\lambda) - 1)  } \big( 2 c_p^2 (\cosh( 2\lambda) - 1)  - e^{-2\lambda} + 1\big) G( \lambda) + \sum_{p \in \pi_+^*} s_p^2 
\end{align}
and thus conclude with the observation $G(0)=1$ at \eqref{eq:mgf-N}.

\subsubsection*{\textbf{Proof of \eqref{eq:mgf-O}:}} Now consider general operators $O$ such that $O_{p,q}\in \ell^2( \pi_+^*) \times \ell^2( \pi_+^*)$. The goal is to prove that \eqref{eq:Fp2} has a unique solution 
\begin{align}
\label{eq:F-space}
F_{p,q}( \lambda) \in \ell^2( \pi_+^*) \times \ell^2( \pi_+^*) \; . 
\end{align}
We remark that \eqref{eq:F-space} ensures that the sum of the r.h.s. of \eqref{eq:derivG2} is finite (since $O_{p,q} \in \ell^2( \pi_+^*) \times \ell^2( \pi_+^*)$ by assumption and $s_p \in \ell^2( \pi_+^*)$, $c_p \in \ell^\infty( \pi_+^*)$ by definition of $\nu_p$ in \eqref{def:sigma}). 

To this end, we first use a fixed point argument to show that \eqref{eq:Fp2} has a unique solution and  second, construct an explicit solution (that, thus, is the unique solution to \eqref{eq:Fp2}). 

\textit{Uniqueness:} We observe that we can write \eqref{eq:Fp2} as 
\begin{align}
F_{p,q} ( \lambda) = (T F( \lambda))_{p,q} 
\end{align}
where the operator 
\begin{align}
\label{def:T}
T: \ell^2( \pi_+^*) \times \ell^2( \pi_+^*) \rightarrow \ell^2( \pi_+^*) \times \ell^2( \pi_+^*)
\end{align}
acts as 
\begin{align}
(T F(\lambda))_{p,q} = A_{p,q} ( \lambda) G( \lambda) + \sum_{k, \ell \in \pi_+^*} D_{p,q,k,\ell} ( \lambda) F_{k,\ell} ( \lambda)  \; . 
\end{align}
and the coefficients $D_{p,q,k,\ell}( \lambda)$ are given by \eqref{def:D-1}. The operator $T$ is well-defined since, on the one hand  
\begin{align}
\vert G( \lambda)  \vert = \vert \langle  \Omega, e^{K_\nu} d\Gamma ( O ) e^{\lambda d \Gamma ( O)}  e^{-K_\nu} \Omega \rangle \vert \leq  \langle \Omega,  e^{ \lambda_1  ( \mathcal{N} + 1)} \Omega \rangle \leq C 
\end{align}
from Lemma \ref{lemma:bogo-standard}, and, on the other hand 
\begin{align}
A_{p,q} ( \lambda), \; \;  \sum_{k,\ell \in \pi_+^*} D_{p,q,k,\ell} (\lambda) F_{k,\ell} \in \ell^2 ( \pi_+^*) \times \ell^2( \pi_+^*) , 
\end{align}
for any $F_{p,q} \in \ell^2 (\pi_+^*) \times \ell^2( \pi_+^*)$ as we prove in the following: 

We start with the properties of $A_{p,q} ( \lambda)$ defined in \eqref{def:Apq-1} and observe that due to some cancellations with the first term of the r.h.s. of \eqref{def:Apq-1}, we have 
\begin{align}
A_{p,q} (\lambda) =& c_p^2c_q s_p \big( e^{\lambda \overline{O}_{p,q}} - \delta_{p, q}\big) \notag \\
&+ c_p c_q \big\langle c   \big(e^{\lambda O_{-p}} - \delta_{\cdot, -p} \big) , s \big( e^{\lambda \overline{O}_q^-}  - \delta_{-\cdot, q}\big) \big\rangle_{\ell^2}  + c_p c_q^2 s_q ( e^{\lambda \overline{O}_{-q,-p}} - \delta_{-p,-q} )  \notag \\
&+ s_p s_q \big\langle s (e^{-\lambda \overline{O}^-_{-p}}, c \big( e^{-\lambda O_q} - \delta_{\cdot,q} \big) \big\rangle + s_p s_q^2 c_q  \big( e^{-\lambda O_{-q,-p}} - \delta_{-p,-q} \big) \notag \\
&- s_p c_q \overline{\big\langle s e^{- \lambda O_{-p}}, s \big( e^{\lambda O_{-q}} - \delta_{\cdot, -q} \big) \big\rangle } -  s_pc_qs_q^2 \big( e^{-\lambda O_{-q,-p}} - \delta_{-p,-q} \big)   \notag \\
&- c_p s_q \overline{\big\langle s e^{-\lambda O_q}, s\big(e^{\lambda O_p} - \delta_{\cdot,p} \big) \big\rangle } - c_ps_q s_p^2 \big( e^{-\lambda O_{p,q}} - \delta_{p,q} \big) \; . \label{eq:A-split}
\end{align}
By Taylor's theorem, we have $e^{\lambda O_{q,p}} = \delta_{p,q} + \lambda O_{q,p} e^{\theta \lambda O_{q,-p}}$ for some $\theta \in \mathbb{R}$ and it follows for the first term of the r.h.s. of  \eqref{eq:A-split} 
\begin{align}
\sum_{p,q \in \pi_+^*} \big\vert c_p^2c_q s_p \big( e^{\lambda \overline{O}_{p,q}} - \delta_{p, q}\big) \big\vert^2 \leq   \lambda^2 \sum_{p,q \in \pi_+^*} c_p^4 c_q^2 s_p^2 O_{p,q}^2 < \infty 
\end{align}
since $\sup_{p,q} \vert O_{p,q} \vert \leq C$ and $\| O \|_{\ell^2 ( \pi_+^*) \times \ell^2( \pi_+^*)} < \infty$. The third, fifth, seventh and ninth term of the r.h.s. of \eqref{eq:A-split} can be bounded similarly. With the same arguments, we find for the second term of the r.h.s. of \eqref{eq:A-split}
\begin{align}
\sum_{p,q \in \pi_+^*}  & \big\vert c_p c_q \big\langle c   \big(e^{\lambda O_{-p}} - \delta_{\cdot, -p} \big) , s \big( e^{\lambda \overline{O}_q^-}  - \delta_{-\cdot, q}\big) \big\rangle_{\ell^2} \big\vert^2 \notag \\
&\leq C \lambda^2 \sum_{p,q,k \in \pi_+^*}  c_k^2 s_k^2 \vert O_{k,-p} \vert^2 \; \vert O_{k,-q} \vert^2  \leq C \sup_{k \in \pi_+^*} s_k^2 c_k^2 \sum_{p \in \pi_+^*} \vert O_{k,p} \vert^2 \sum_{k,q \in \pi_+^*} \vert O_{k,q} \vert^2 < \infty 
\end{align}
where we substituted $p,q$ with $-p$ resp. $-q$ and used that, by definition, $s_k = s_{-k}, c_k = c_{-k}$ . The forth, sixth and eight term of the r.h.s. of \eqref{eq:A-split} are bounded by the same arguments. Hence, we arrive hat 
\begin{align}
\| A ( \lambda) \|_{\ell^2( \pi_+^*) \times \ell^2 ( \pi_+^*)} \leq C \lambda \; \label{eq:A-estimate}
\end{align}
and thus, $A_{p,q} ( \lambda) \in \ell^2( \pi_+^*) \times \ell^2( \pi_+^*)$. 


Next we show that $\sum_{k,\ell \in \pi_+^*} D_{p,q,k,\ell} (\lambda) F_{k,\ell}  \in \ell^2( \pi_+^*)\times \ell^2 ( \pi_+^*) $ for any $F_{p,q} \in \ell^2 (\pi_+^*) \times \ell^2( \pi_+^*)$. For this, we observe that by a cancellation between the first and the forth, as well as between the second and the third term of the sum , we have 
\begin{align}
&D_{p,q,k,\ell} ( \lambda) \notag \\
&\quad = c_p c_k  c_q  s_\ell \tau_\ell \bigg(   \big( e^{\lambda \overline{O}_{-k,-p}}  - \delta_{k,p} \big) \big( e^{\lambda \overline{O}_{\ell,q}} - \delta_{\ell,q} \big) + \delta_{k,p}  \big( e^{\lambda \overline{O}_{\ell,q}} - \delta_{\ell,q} \big) + \big( e^{\lambda \overline{O}_{-k,-p}} - \delta_{k,p}) \delta_{\ell,q} \notag \\
 &\hspace{2.5cm}  +  \big(e^{-\lambda O_{k,-p}} - \delta_{k,-p}\big) \big( e^{-\lambda O_{-\ell,q}}  - \delta_{-\ell,q} \big) + \delta_{k,-p}  \big( e^{-\lambda O_{-\ell,q}} - \delta_{-\ell,q} \big) \notag \\
 &\hspace{2.5cm}+ \big(e^{-\lambda O_{k,-p}} - \delta_{k,-p}\big)  \delta_{-\ell,q}  - \big(e^{-\lambda O_{k,-p}} - \delta_{k,-p}\big)  \big( e^{\lambda \overline{O}_{\ell, -q}} - \delta_{\ell,-q} \big)  \notag \\
 &\hspace{2.5cm}- \delta_{k,-p} \big(e^{\lambda \overline{O}_{\ell, -q}}  - \delta_{\ell,q} \big) - \big(e^{-\lambda O_{k,-p}} - \delta_{k,-p}\big)  \delta_{-\ell,q} \notag \\
  &\hspace{2.5cm}  - \big( e^{-\lambda O_{k,p}} - \delta_{k,p} \big) \big( e^{\lambda \overline{O}_{\ell, q}} - \delta_{\ell,q} \big) - \delta_{k,p} \big( e^{\lambda \overline{O}_{\ell, q}} - \delta_{\ell,q} \big) \bigg) -  \big( e^{-\lambda O_{k,p}} - \delta_{k,p} \big) \delta_{\ell,q} \bigg) 
\end{align}
and thus 
\begin{align}
\sum_{k, \ell \in \pi_+^*}  & D_{p,q,k,\ell} ( \lambda) F_{k,\ell} \notag \\
 =& c_p  s_q^2 \sum_{k \in \pi_+^*}  c_k   \bigg[ \big( e^{\lambda \overline{O}_{-k,-p}} - \delta_{k,p}) -   \big( e^{-\lambda O_{k,p}} - \delta_{k,p} \big) \bigg] F_{k,q} \notag \\
&+ c_p s_q^2 \sum_{k \in \pi_+^*}  c_k \bigg[ \big(e^{-\lambda O_{k,-p}} - \delta_{k,-p}\big) - \big(e^{-\lambda O_{k,-p}} - \delta_{k,-p}\big)   \bigg] F_{k,-q}  \notag \\
&+ c_p^2 c_q  \sum_{\ell \in \pi_+^*} s_\ell \tau_\ell \bigg[ \big( e^{\lambda \overline{O}_{\ell,q}} - \delta_{\ell,q} \big) - \big( e^{\lambda \overline{O}_{\ell, q}} - \delta_{\ell,q} \big) \bigg) \bigg] F_{p,\ell}\notag \\
&+ c_p^2 c_q  \sum_{\ell \in \pi_+^*} s_\ell \tau_\ell  \bigg[  \big( e^{-\lambda O_{-\ell,q}} - \delta_{-\ell,q} \big) - \big(e^{\lambda \overline{O}_{\ell, -q}}  - \delta_{\ell,q} \big)\bigg] F_{-p,\ell} \notag \\
&+ c_pc_q\sum_{k, \ell \in \pi_+^*} c_ks_\ell \tau_\ell \bigg[  \big( e^{\lambda \overline{O}_{-k,-p}}  - \delta_{k,p} \big) \big( e^{\lambda \overline{O}_{\ell,q}} - \delta_{\ell,q} \big) \notag \\
& \hspace{5cm} - \big( e^{-\lambda O_{k,p}} - \delta_{k,p} \big) \big( e^{\lambda \overline{O}_{\ell, q}} - \delta_{\ell,q} \big) \bigg] F_{k, \ell} \notag \\
&+ c_pc_q\sum_{k, \ell \in \pi_+^*} c_ks_\ell \tau_\ell \bigg[  \big(e^{-\lambda O_{k,-p}} - \delta_{k,-p}\big) \big( e^{-\lambda O_{-\ell,q}}  - \delta_{-\ell,q} \big)\notag \\
& \hspace{5cm} -  \big(e^{-\lambda O_{k,-p}} - \delta_{k,-p}\big)  \big( e^{\lambda \overline{O}_{\ell, -q}} - \delta_{\ell,-q} \big) \bigg] F_{k, \ell} \label{eq:D-decomp}
\end{align}
whose terms we now estimate separately. For the first term of the r.h.s. of \eqref{eq:D-decomp} we find with Cauchy Schwarz 
\begin{align}
\sum_{p,q \in \pi_+^*} & s_q^4 \bigg\vert \sum_{k \in \pi_+^*}  \bigg[ \big( e^{\lambda \overline{O}_{-k,-p}} - \delta_{k,p}) -   \big( e^{-\lambda O_{k,p}} - \delta_{k,p} \big) \bigg] F_{k,q} \bigg\vert^2 \notag \\
& \leq \sum_{p,q \in \pi_+^*} s_q^4  \sum_{k_1 \in \pi_+^*}  \bigg[ \big( e^{\lambda \overline{O}_{-k_1,-p}} - \delta_{k_1,p}) -   \big( e^{-\lambda O_{k_1,p}} - \delta_{k_1,p} \big) \bigg]^2 \; \sum_{k_2 \in \pi_+^*} F_{k_2,q}^2 
\end{align}
By Taylor's theorem, we have $e^{\lambda \overline{O}_{-k,-p}} - \delta_{k,p} = \lambda \overline{O}_{-k,-p} e^{\theta_1 \overline{O}_{-k,-p}}$ resp. $e^{-\lambda O_{k,p}} - \delta_{k,p} = \lambda  O_{k,p} e^{\theta_2 O_{k,p}}$ for some $\theta_1, \theta_2 \in \mathbb{R}$, and since $ \| O \|_{\ell^2 ( \pi_+^*)\times \ell^2( \pi_+^*)} < \infty $ by assumption yielding $  \| O \|_{\ell^\infty \times \ell^\infty} < \infty$, we find 
\begin{align}
\sum_{p,q \in \pi_+^*} & s_q^4 \bigg\vert \sum_{k \in \pi_+^*}  \bigg[ \big( e^{\lambda \overline{O}_{-k,-p}} - \delta_{k,p}) -   \big( e^{-\lambda O_{k,p}} - \delta_{k,p} \big) \bigg] F_{k,q} \bigg\vert^2  \leq C \lambda^2 \| F \|_{\ell^2 ( \pi_+^*)\times \ell^2( \pi_+^*)}^2 \label{eq:D-decomp-bound1} 
\end{align}
where we used that $\sup_{q \in \pi_+^*} s_q^2 < \infty$ from \eqref{def:sigma}. The second, third and forth term of the r.h.s. of \eqref{eq:D-decomp} can be estimated similarly arriving at a similar bound as \eqref{eq:D-decomp-bound1}. For the fifth term of the r.h.s. of \eqref{eq:D-decomp} we recall that with the same arguments as before 
\begin{align}
\sum_{p,q \in \pi_+^*}  & \bigg\vert  \sum_{k, \ell \in \pi_+^*} \bigg[  c_ks_\ell \tau_\ell \bigg[  \big( e^{\lambda \overline{O}_{-k,-p}}  - \delta_{k,p} \big) \big( e^{\lambda \overline{O}_{\ell,q}} - \delta_{\ell,q} \big) - \big( e^{-\lambda O_{k,p}} - \delta_{k,p} \big) \big( e^{\lambda \overline{O}_{\ell, q}} - \delta_{\ell,q} \big) \bigg] F_{k, \ell}  \bigg\vert^2  \notag \\
 =&  \lambda^4  \sum_{p,q \in \pi_+^*} \bigg\vert  \sum_{k, \ell \in \pi_+^*} c_ks_\ell \tau_\ell  \bigg[ \overline{O}_{-k,-p} \overline{O}_{\ell,q} e^{\theta_1 \lambda \overline{O}_{-k,-p}}  e^{\theta_2 \lambda \overline{O}_{\ell,q}}   + O_{k,p} \overline{O}_{\ell,q} e^{\theta_3 \lambda O_{k,p}}  e^{\theta_4 \lambda \overline{O}_{\ell,q}} \bigg] F_{k, \ell}  \bigg\vert^2 
\end{align}
for some constants $\theta_i \in \mathbb{R}$ and we conclude using that $\sup_{q \in \pi_+^*} c_q \geq 1$ by Cauchy Schwarz and similar arguments as before 
\begin{align}
\sum_{p,q \in \pi_+^*}  & \bigg\vert  \sum_{k, \ell \in \pi_+^*} \bigg[  c_ks_\ell \tau_\ell \bigg[  \big( e^{\lambda \overline{O}_{-k,-p}}  - \delta_{k,p} \big) \big( e^{\lambda \overline{O}_{\ell,q}} - \delta_{\ell,q} \big) - \big( e^{-\lambda O_{k,p}} - \delta_{k,p} \big) \big( e^{\lambda \overline{O}_{\ell, q}} - \delta_{\ell,q} \big) \bigg] F_{k, \ell}  \bigg\vert^2 \notag \\
 \leq&  C  \lambda^4 \| F \|_{\ell^2 ( \pi_+^*) \times \ell^2 ( \pi_+^*)}^2 \; . 
\end{align}
The estimate for the remaining term of the r.h.s. of \eqref{eq:D-decomp} follows in the same way. We thus arrive for sufficiently small $\vert \lambda\vert $ at 
\begin{align}
\sum_{p,q \in \pi_+^*} \bigg\vert  \sum_{k, \ell \in \pi_+^*}  & D_{p,q,k,\ell} ( \lambda) F_{k,\ell}  \bigg\vert^2 \leq C \lambda^2 \| F \|_{\ell^2( \pi_+^*) \times \ell^2( \pi_+^*)}^2  \label{eq:estimate-D}
\end{align}
leading to the desired conclusion that $\sum_{k, \ell \in \pi_+^*}  D_{p,q,k,\ell} ( \lambda) F_{k,\ell}  \in \ell^2 ( \pi_+^*)\times \ell^2( \pi_+^*) $ for $F_{k,\ell} \in \ell^2 ( \pi_+^*) \times \ell^2 ( \pi_+^*)$, and in particular that the operator $T$ defined in \eqref{def:T} is well defined. 

Next, we aim to show that the operator $T$ defined in \eqref{def:T} has a unique fix point, and, thus, consequently that \eqref{eq:Fp2} has a unique solution in $\ell^2 ( \pi_+^*)\times \ell^2( \pi_+^*)$. For this, by Banach's fixed point theorem, it suffices to show that the operator $T$ is a contraction. In fact, as an immediate consequence of \eqref{eq:estimate-D} we have for $F^{(1)}, F^{(2)} \in \ell^2( \pi_+^*) \times \ell^2( \pi_+^*)$
\begin{align}
\| \big( T F^{(1)}  \big) ( \lambda) - \big( T  F^{(2)} \big) ( \lambda) \|_{\ell^2 ( \pi_+^*)\times \ell^2( \pi_+^*)}^2  =& \sum_{p,q \in \pi_+^*} \big\vert D_{p,q,k, \ell} ( \lambda) \big[ F^{(1)}_{k, \ell} - F^{(2)}_{k, \ell} \big] \bigg\vert^2 \notag \\
\leq& C \vert \lambda \vert \| F^{(1)} - F^{(2)} \|_{\ell^2( \pi_+^*) \times \ell^2( \pi_+^*)}^2 
\end{align}
and we finally conclude that for sufficiently small $\vert \lambda \vert $, the operator $T$ is a contraction and, therefore \eqref{eq:Fp2} has a unique solution in $\ell^2( \pi_+^*) \times \ell^2( \pi_+^*)$. 

\textit{Explicit solution:} We claim that 
\begin{align}
\widetilde{F}_{p,q} ( \lambda) = \sum_{k,\ell \in \pi_+^*}\sum_{j=0}^\infty \big( D_{p,q,k,\ell} ( \lambda) \big)^j A_{k, \ell} ( \lambda)  G( \lambda) \label{eq:F-solution}
\end{align}
solves \eqref{eq:Fp2} where we introduced the notation $\big( D_{p,q,k,\ell} ( \lambda)\big)^j$ for the operator that is recursively defined through 
\begin{align}
\big( D_{p,q,k,\ell} ( \lambda) \big)^0 = \delta_{p,k} \delta_{q,\ell} 
\end{align}
and 
\begin{align}
\label{eq:power-D}
\big( D_{p,q,k,\ell} ( \lambda) \big)^{j+1} = \sum_{m,n \in \pi_+^*}  D_{p,q,m,n} ( \lambda) \big( D_{m,n,k,\ell} ( \lambda) \big)^{j} \; . 
\end{align}
The claim follows by standard arguments for solution to integral equations, that we present here for consistency: 

First, note that the series on the r.h.s. of \eqref{eq:F-solution} converges: By definition \eqref{eq:power-D}, we have for $j\geq 1 $
\begin{align}
\sum_{p,q \in \pi_+^*}  & \bigg\vert \sum_{k, \ell \in \pi_+^*}\big( D_{p,q,k,\ell} ( \lambda) \big)^j A_{k, \ell} ( \lambda )\bigg\vert^2 \notag \\
=&   \sum_{p,q \in \pi_+^*}  \bigg\vert \sum_{m,n \in \pi_+^*}D_{p,q,m,n} ( \lambda) \sum_{k, \ell \in \pi_+^*}\big( D_{m,n,k,\ell} ( \lambda) \big)^{j-1}  A_{k, \ell} ( \lambda) \bigg\vert^2 
\end{align}
and thus from \eqref{eq:estimate-D} 
\begin{align}
\sum_{p,q \in \pi_+^*}  & \bigg\vert \sum_{k, \ell \in \pi_+^*}\big( D_{p,q,k,\ell} ( \lambda) \big)^j A_{k, \ell} ( \lambda) \bigg\vert^2 \notag \\
\leq& C \vert \lambda\vert    \sum_{m,n \in \pi_+^*} \bigg\vert \sum_{k, \ell \in \pi_+^*}\big( D_{m,n,k,\ell} ( \lambda) \big)^{j-1}  A_{k, \ell}\bigg\vert^2  \; . 
\end{align}
Iterating this step we thus arrive at 
\begin{align}
\sum_{p,q \in \pi_+^*}  & \bigg\vert \sum_{k, \ell \in \pi_+^*}\big( D_{p,q,k,\ell} ( \lambda) \big)^j A_{k, \ell}( \lambda) \bigg\vert^2 
\leq C^j \vert\lambda\vert^j  \| A \|_{\ell^2 ( \pi_+^*)\times \ell^2( \pi_+^*)}^2 \leq C^j \vert\lambda\vert^{j+1} 
\end{align}
where we concluded by \eqref{eq:A-estimate}. Thus, from \eqref{eq:F-solution} we find 
\begin{align}
\| \widetilde{F} ( \lambda) \|_{\ell^2( \pi_+^*) \times \ell^2( \pi_+^*)} \leq C_0 \vert\lambda\vert \sum_{j=0}^\infty ( C\vert\lambda\vert)^j  \label{eq:series-conv}
\end{align}
for some constants $C_0,C>0$ that is a finite series for sufficiently small $\vert\lambda\vert$. 

It remains to show, that $\widetilde{F}_{p,q} ( \lambda)$ given by \eqref{eq:F-solution} indeed solves \eqref{eq:Fp2}. For this we write using \eqref{eq:power-D} and \eqref{eq:series-conv}
\begin{align}
\widetilde{F}_{p,q} ( \lambda) =& \sum_{k,\ell \in \pi_+^*}\sum_{j=0}^\infty \big( D_{p,q,k,\ell} ( \lambda) \big)^j A_{k, \ell} ( \lambda) G( \lambda)  \notag \\
=& \sum_{k,\ell \in \pi_+^*}\lim_{n \rightarrow \infty} \sum_{j=0}^n \big( D_{p,q,k,\ell} ( \lambda) \big)^j A_{k, \ell} ( \lambda) G( \lambda)\notag \\
=& \lim_{n \rightarrow \infty} \bigg[ \sum_{k,\ell \in \pi_+^*}  \sum_{k_1,\ell_1 \in \pi_+^*}D_{p,q,k_1,\ell_2} ( \lambda) \sum_{j=1}^n \big( D_{k_1,\ell_1,k,\ell} ( \lambda) \big)^{j-1} A_{k, \ell} ( \lambda) G( \lambda) + A_{p,q} ( \lambda) G( \lambda) \bigg] \notag \\\
=& \lim_{n \rightarrow \infty} \bigg[ \sum_{k,\ell \in \pi_+^*}\sum_{k_1,\ell_1 \in \pi_+^*}D_{p,q,k_1,\ell_2} ( \lambda) \sum_{j=0}^{n-1} \big( D_{k_1,\ell_1,k,\ell} ( \lambda) \big)^{j} A_{k, \ell} ( \lambda) G( \lambda) + A_{p,q} ( \lambda) G( \lambda) \bigg] \notag \\\
=& \sum_{k,\ell \in \pi_+^*} \sum_{k_1,\ell_1 \in \pi_+^*}D_{p,q,k_1,\ell_2} ( \lambda) \lim_{n \rightarrow \infty}  \sum_{j=0}^{n-1} \big( D_{k_1,\ell_1,k,\ell} ( \lambda) \big)^{j} A_{k, \ell} ( \lambda) G( \lambda) + A_{p,q} ( \lambda) G( \lambda) \notag \\
=& \sum_{k,\ell \in \pi_+^*} \sum_{k_1,\ell_1 \in \pi_+^*}D_{p,q,k_1,\ell_2} ( \lambda) \widetilde{F}_{k_1,\ell_1 \in \pi_+^*}  + A_{p,q}( \lambda)  G( \lambda)
\end{align} 
and it follows that $\widetilde{F}_{p,q} ( \lambda)$ solves \eqref{eq:Fp2}. 

Summarizing, we have proven that $\widetilde{F}_{p,q} ( \lambda) \in \ell^2( \pi_+^*) \times \ell^2 ( \pi_+^*)$ given by \eqref{eq:F-solution} denotes the unique solution to \eqref{eq:Fp2}. Then, it follows from \eqref{eq:G1} together with \eqref{eq:derivG2} that 
\begin{align}
G'(\lambda) =& \sum_{p,q \in \pi_+^*} s_p c_q O_{p,q}\widetilde{F}_{p,q} ( \lambda)  + \sum_{p \in \Lambda_+^*} s_p^2 O_{p,p} \notag \\
=& \sum_{p,q,k, \ell \in \pi_+^*} s_p c_q O_{p,q} \big( D_{p,q,k,\ell} ( \lambda) \big)^j A_{k, \ell} ( \lambda) G( \lambda)  + \sum_{p \in \Lambda_+^*} s_p^2 O_{p,p}  
\end{align}
Therefore noticing that $G( 0) =1$ we have 
\begin{align}
\label{eq:final-G-O}
G( \lambda ) = \exp \bigg( \int_0^\lambda  \sum_{p,q, k \ell \in \pi_+^*} s_p c_q O_{p,q} \sum_{j=0}^\infty \big( D_{p,q,k,\ell} ( \lambda) \big)^j A_{k, \ell} ( \lambda ) d \lambda  + \lambda \sum_{p \in \Lambda_+^*} s_p^2 O_{p,p}  \bigg) 
\end{align}
and the integral is finite for sufficiently small $\vert \lambda \vert $  following from \eqref{eq:series-conv}.

\end{proof}

\section*{Acknowledgements} SR thanks Christian Brennecke and Phan Th\`anh Nam for helpful discussions. 

\section*{Declarations}

\subsection*{Funding}

SR is supported by the European Research Council via the ERC CoG RAMBAS - Project - Nr. 10104424.

\subsection*{Data availability}

Data sharing not applicable to this article as no datasets were generated or analyzed.

\subsection*{Conflict of interest}

The author has no competing interests to declare that are relevant to the content of this article.


\begin{thebibliography}{29}


\bibitem{WC-95} M. H. Anderson, J. R. Ensher, M. R. Matthews, C. E. Wieman, and E. A. Cornell. 
Observation of Bose-Einstein condensation in a dilute atomic vapor. {\em Science} 269, pp. 198--201, (1995). 




\bibitem{BCCS_cond} 
C. Boccato, C. Brennecke, S. Cenatiempo and B. Schlein. Complete {B}ose-{E}instein condensation in the {G}ross-{P}itaevskii regime. {\em Commun. Math. Phys.} 359, pp. 975--1026, (2018). 



\bibitem{BCCS} C. Boccato, C. Brennecke, S. Cenatiempo and B. Schlein. Bogoliubov Theory in the {G}ross-{P}itaevskii Limit. {\em Acta Mathematica} 222, pp.  219--335, (2019). 

\bibitem{BCCS_optimal} 
C. Boccato, C. Brennecke, S. Cenatiempo and B. Schlein. Optimal Rate for Bose-Einstein Condensation in the Gross-Pitaevskii Regime. {\em Commun. Math. Phys.} 376, pp. 1311--1395,  (2020).  

\bibitem{Bogoliubov-47} N. N. Bogoliubov. On the theory of superfluidity. {\em J. Phys. (USSR)} 11, pp. 23--32, (1947). 


\bibitem{Bose-24} S. Bose. Plancks Gesetz und Lichtquantenhypothese. {\em Z. Phys.} 26, pp. 178--181, (1924). 

\bibitem{BKS} G. Ben Arous, K. Kirkpatrick, B. Schlein, Central Limit Theorem in Many-Body Quantum Dynamics. {\em Commun. Math. Phys}, 321, pp. 371--417, (2013). 

\bibitem{BS} C. Brennecke and B. Schlein. Gross-Pitaevskii Dynamics for Bose-Einstein Condensates. {\em Analysis \& PDE} 12, pp. 1513-1596, (2019). 

\bibitem{BSS_optimal} 
C. Brennecke, B. Schlein and S. Schraven. Bose-{E}instein Condensation with Optimal Rate for Trapped Bosons in the {G}ross-{P}itaevskii Regime. {\em Math. Phys. Anal. Geom.} 25:12, (2022). 

\bibitem{BSS_bogo}
C. Brennecke, B. Schlein, and S. Schraven. Bogoliubov Theory for Trapped Bosons in the {G}ross-{P}itaevskii Regime. {\em Ann. Henri Poincar\'{e}} 23, pp. 1583--1658, (2022).  

\bibitem{COS} C. Caraci, J. Oldenburg, B. Schlein, Quantum Fluctuations of Many-Body Dynamics around the Gross-Pitaevskii Equation, \textit{Ann. Inst. H. Poincar\'{e} C Anal. Non Lin\'{e}aire}, published online first, (2024) . 

\bibitem{Einstein-24} A. Einstein. Quantentheorie des einatomigen idealen Gases.  {\em Sitzungsberichte der
Preu\ss ischen Akademie der Wissenschaften} XXII (1924), pp. 261--267.


\bibitem{K-95}  K. B. Davis, M. O. Mewes, M. R. Andrews, N. J. van Druten, D. S. Durfee, D. M.
Kurn, and W. Ketterle. Bose-Einstein Condensation in a Gas of Sodium Atoms. {\em Phys.
Rev. Lett.} 75, pp. 3969--3973, (1995). 

\bibitem{H} C. Hainzl. Another Proof of BEC in the {GP}-limit. \textit{J. Math. Phys}. 62 (5), pp. 459-485, (2021).

\bibitem{HST} C. Hainzl, B. Schlein and A. Triay. Bogoliubov theory in the {G}ross-{P}itaevskii limit: a simplified approach. {\em Forum Math. Sigma} 10 (e10), (2022).

\bibitem{KRS} K.~Kirkpatrick, S.~Rademacher, and B.~Schlein. A large deviation principle for many--body quantum dynamics. {\em Ann. Henri Poincar\'e}  22, pp. 2595--2618, (2021). 


\bibitem{LS-02} 
E.H. Lieb and R. Seiringer. Proof of Bose-Einstein Condensation for Dilute Trapped Gases. {\em Phys. Rev. Lett.} 88, p. 170409, (2002). 


\bibitem{LS} 
E.H. Lieb, and R. Seiringer, Derivation of the {G}ross-{P}itaevskii equation for rotating Bose gases. {\em Commun. Math. Phys.} 264, pp. 505--537, (2006). 


\bibitem{LNSS} 
M. Lewin, P.T. Nam, S. Serfaty and J.P. Solovej. Bogoliubov spectrum of
interacting Bose gases. {\em Comm. Pure Appl. Math.} 68, pp. 413--471, (2014). 

\bibitem{N21} P.T. Nam, Bogoliubov excitation spectrum of Bose gases, In: A.  Hujdurovi\'{c}, K. Kutnar, D. Marusic, S. Miklavic, T. Pisanski, P. Sparl (eds): European Congress of Mathematics, \textit{EMS Press}, (2021). 


\bibitem{NNRT} 
P.T. Nam, M. Napi\'{o}rkowski, J. Ricaud, and A. Triay. Optimal rate of condensation for trapped bosons in the Gross-Pitaevskii regime. {\em Analysis \& PDE} 15, pp.  1585--1616, (2022). 

\bibitem{NR}
P.T. Nam, S. Rademacher, Exponential bounds of the condensation for dilute Bose gases, \textit{Trans. Amer. Math. Soc.}, 378, pp. 3229-3278, (2025). 

\bibitem{NRS} 
P.T. Nam, N. Rougerie and  R. Seiringer.  Ground states of large bosonic systems: The {G}ross-{P}itaevskii limit revisited. {\em Analysis \& PDE}  9, pp. 459--485, (2016). 


\bibitem{NT} P.T. Nam and  A. Triay. {B}ogoliubov excitation spectrum of trapped Bose gases in the {G}ross-{P}itaevskii regime. {\em J. Math. Pures Appl.} (to appear), Preprint: arXiv:2106.11949 (2021).  


\bibitem{Nap18} M. Napi\'{o}rkowski, Recent Advances in the Theory of Bogoliubov Hamiltonians. In: D. Cadamuro, M. Duell, W. Dybalski, S. Simonella, S. (eds): Macroscopic Limits of Quantum Systems. MaLiQS 2017, \textit{Springer Proceedings in Mathematics \& Statistics}, vol. 270, (2018)


\bibitem{PRV} L. Portinale, S. Rademacher, D. Virosztek, Limit theorems for empirical measures of interacting quantum systems in Wasserstein space, Preprint: arXiv:2312.0054, (2023). 

\bibitem{Rsing} S. Rademacher, Central limit theorem for Bose gases interacting through singular potentials. \textit{Lett. Math. Phys.} 110, 2143-2174 (2020).

\bibitem{R} S.~Rademacher. Large deviations for the ground state of weakly interacting {B}ose gases. {\em Ann. Henri Poincar\'e}, (2024). 


\bibitem{RS} S. Rademacher, B. Schlein. Central limit theorem for Bose-Einstein condensates. \textit{J. of Math. Phys.}, 60(7):071902,  (2019). 

\bibitem{RSe}
S.~Rademacher and R.~Seiringer. Large deviation estimates for weakly interacting bosons. {\em J. Stat. Phys.} 188  (9), (2022). 

\bibitem{S22} B. Schlein, Bose gases in the Gross-Pitaevskii limit: A survey of some rigorous results, In: R. Frank, A. Laptev, M. Lewin, R. Seiringer (eds): The Physics and Mathematics of Elliott Lieb, \textit{EMS Press}, vol. II, (2022). 

\bibitem{Simon} B. Simon, The $P( \Phi)_2$ Euclidean Quantum Field Theory, \textit{Princeton Series in Physics}, 1974. 

\bibitem{SolovejLN} J. P. Solovej, Lecture notes on: Many Body Quantum Mechanics, available online \href{https://web.math.ku.dk/solovej/MANYBODY/mbnotes-ptn-5-3-14.pdf}{https://web.math.ku.dk/solovej/MANYBODY/mbnotes-ptn-5-3-14.pdf}.  

\end{thebibliography}
\end{document}